\documentclass[9pt]{amsart}
\usepackage{amscd,amssymb,epsfig}
\usepackage{graphicx}

\newtheorem{theorem}{Theorem}[section]

\newtheorem{lemma}[theorem]{Lemma}
\newtheorem{fact}[theorem]{Fact}
\newtheorem{proposition}[theorem]{Proposition}
\newtheorem{corollary}[theorem]{Corollary}
\newtheorem{result}[theorem]{Result}
\theoremstyle{definition}

\theoremstyle{definition}
\newtheorem{construction}[theorem]{Construction}
\theoremstyle{remark}
\newtheorem{remark}[theorem]{Remark}
\numberwithin{figure}{section}
\numberwithin{table}{section}

\input epsf

\begin{document}

\title[Fatgraph Models of Proteins]
{Fatgraph models of proteins}

\author{R. C. Penner}
\address{Departments of Mathematics and Physics/Astronomy\\
University of Southern California\\
Los Angeles, CA 90089\\
USA\\
~{\rm and}~Center for the Topology and Quantization of Moduli Spaces\\
Department of Mathematics\\
Aarhus University\\
DK-8000 Aarhus C, Denmark\\}
\email{rpenner{\char'100}math.usc.edu}

\author{Michael Knudsen}
\address{Bioinformatics Research Center\\
Aarhus University\\
DK-8000 Aarhus C, Denmark\\}
\email{micknudsen{\char'100}gmail.com}

\author{Carsten Wiuf}
\address{Bioinformatics Research Center\\
and Centre for Membrane Pumps in Cells and Disease--PUMPKIN\\
Danish National Research Foundation\\
Aarhus University\\
DK-8000 Aarhus C, Denmark\\}
\email{wiuf{\char'100}birc.au.dk}

\author{J{\o}rgen Ellegaard Andersen}
\address{Center for the Topology and Quantization of Moduli Spaces\\
Department of Mathematics\\
Aarhus University\\
DK-8000 Aarhus C, Denmark\\}
\email{andersen{\char'100}imf.au.dk}

\keywords{protein descriptors, protein structure, fatgraphs}

\thanks{It is a pleasure for RCP to thank
the Preuss Foundation for funding an exploratory conference in 2005
at the University of Southern California,
immediately after which a precursor of the model described here was derived,
and to thank the participants of that conference, Ken Dill, Soren Istrail, Hubert Saleur,
Arieh Warshel, Michael Waterman, and especially Alexei Finkelstein, for
stimulating and provocative comments.}

\thanks{MK is supported by the Center for Theory in the Natural Sciences at Aarhus University.}

\thanks{CW is partially supported by the Danish Research Councils.}

\thanks{Excellent computer programming assistance was provided by  Piotr Karasinski
funded by the Center for Theory in Natural Sciences, Aarhus University, and the Danish Ministry
of Science, Technology, and Innovation.}

\thanks{The methods disclosed in this paper are protected by the U.S. provisional patent filing 61/077,277 (July 1 2008) and the Danish priority
application PA 2008 01009 (July 17 2008).}

\begin{abstract} We introduce a new model of proteins, which extends and enhances the
traditional graphical representation by associating a combinatorial object called a fatgraph to any protein based upon its intrinsic geometry.  Fatgraphs can 
easily be stored and manipulated as triples of permutations, and these methods are therefore amenable to fast computer implementation.  Applications include the 
refinement of structural protein classifications and the prediction of geometric and other properties of proteins from their chemical structures.
\end{abstract}

\maketitle

\section*{Introduction}

A ``fatgraph'' $G$ is a graph in the usual sense of the term together with cyclic orderings on the half-edges about each vertex (cf.\ 
Section~\ref{fatgraphsequiv} for the precise definition).  They arose in mathematics \cite{Penner88} as the combinatorial objects indexing orbi-cells in a 
certain decomposition of Riemann's moduli space \cite{Penner88,Strebel} and in physics \cite{BIZ,Hooft} as index sets for the large $N$ limit of certain matrix 
models.  A basic geometric point is that a fatgraph $G$ uniquely determines a corresponding surface $F(G)$ with boundary which contains  $G$ as a deformation 
retract.  Fatgraphs have already proved useful in geometry \cite{Harer-Zagier,Igusa,Mondello,Penner88}, in theoretical physics \cite{Brezinetal,Kontsevich92}, 
and in modeling RNA secondary structures \cite{Penner-Waterman}, for example.

A ``protein'' $P$ is a linear polymer of amino acids (cf.\ Section~\ref{polypeptides} for more precision), and their study is a central theme in contemporary 
biophysics
\cite{Folding,Finkelstein}.  Our main achievement in this paper is to introduce a model of proteins which naturally associates a fatgraph $G(P)$ to a protein 
$P$ based upon the spatial locations of its constituent atoms.  The idea is that
the protein is roughly described geometrically as the concatenation of a sequence of planar polygons called ``peptide units'' meeting at tetrahedral angles at 
pairs of vertices and twisted by pairs of dihedral angles
between the polygons. To each peptide unit, we associate a positively oriented orthonormal 3-frame and a fatgraph building block, and we concatenate these 
building blocks using these 3-frames in a manner naturally determined by the geometry of the Lie group $SO(3)$.  There are furthermore ``hydrogen bonds''  
between atoms contained in the peptide units, and these
are modeled by including further edges connecting the building blocks so as to determine a well-defined fatgraph $G(P)$ from $P$.  Thus, the fatgraph $G(P)$ 
derived from the protein $P$
captures the geometry of the protein ``backbone'' and the geometry and combinatorics of the hydrogen bonding along the backbone; elaborations of this basic 
model are also described
which capture further aspects of protein structure.

The key point is that topological or geometric properties of the fatgraph $G(P)$ can be taken as properties or ``decriptors'' of the protein $P$ itself.
A fundamental aspect not usually relevant in applying fatgraphs is that this construction of $G(P)$ is based on actual experimental data about $P$ in which
there are uncertainties and sometimes errors as well.  Furthermore, the notion that the protein $P$ is comprised of atoms at fixed relative spatial locations, 
which is the basic input to our model,
is itself a biological idealization of the reality that a given protein at equilibrium may have several closely related co-existing geometric incarnations.
In order that the protein descriptors arising from fatgraphs are meaningful characteristics of proteins in light of these remarks, we shall be forced to go 
beyond the usual situation and consider fatgraphs $G$ whose corresponding surfaces $F(G)$ are non-orientable.
This is easily achieved combinatorially by including in the definition of a fatgraph also a coloring of its edges by a set with two elements.

The desired result of ``robust'' protein descriptors, i.e., properties of $G(P)$ that do not change much under small changes in the relative spatial locations 
of the atoms constituting $P$, is a key attribute
of our construction; for example, the number of boundary components and the Euler characteristic of $F(G(P))$ are such robust invariants, and we give a plethora of further 
numerical and non-numerical examples.  Another key point of our construction rests on the fact that biophysicists {\sl already} often associate a graph to a 
protein $P$ based upon its hydrogen and chemical bonding, and our model succeeds in
reproducing this usual graphical depiction of a protein but now with its enhanced structure as a fatgraph $G(P)$, i.e., the graph underlying $G(P)$ is the one
usually associated to $P$ in biophysics.  Furthermore, an important practical point is that fatgraphs can be conveniently stored and manipulated on the computer 
as triples of permutations.

Since this is a math paper whose central purpose is to introduce fatgraph models of proteins, we shall not dwell on biophysical applications but nevertheless 
feel compelled to include here
several such applications as follows.  Certain proteins decompose naturally into ``domains'', roughly 115,000 of which have so far been determined 
experimentally and categorized into
several thousand classes, cf. \cite{CATH,SCOP, DALI,PFAM}.   Our most basic robust descriptors of a domain $P$ are given by the topological types of the surface 
$F(G(P))$ computed with various thresholds of potential energy imposed on the hydrogen bonds (see Section \ref{generalmodel} for details).
We show here that the topological types of $F(G(P))$ for several  such potential energy thresholds uniquely determine $P$ among all known protein globules.
Other such ``injectivity results'' for globules based on various robust protein descriptors are also presented.

This paper is organized as follows.  Section~\ref{polypeptides} introduces an abstract definition of  ``polypeptides'', which give a precise mathematical 
formulation of the biophysics of a protein required for our model; a more detailed discussion of proteins from first principles is given in the beautiful book 
\cite{Finkelstein}, which we heartily recommend.  Section~\ref{fatgraphs} introduces the notion of fatgraphs required here, whose corresponding surfaces may be 
non-orientable, and contains basic results about them.  In particular, a number of results, algorithms, and constructions are presented showing that our methods 
are amenable to fast computer implementation.

Section~\ref{model} is the heart of the paper and describes the fatgraph associated to a polypeptide structure in detail.  Background on $SO(3)$ graph 
connections
is given in Section~\ref{connections}, and this is applied in Section~\ref{backbonesec}, where we explain how the fatgraph building blocks associated with 
peptide units are
concatenated.  Section~\ref{hydrogen} discusses the addition of edges corresponding to hydrogen bonds, thus completing the basic construction of the fatgraph 
model
of a polypeptide structure.  Section~\ref{generalmodel} discusses this basic model and its natural generalizations and extensions for proteins and beyond.  An 
alternative description of this model, which is more physically transparent but less mathematically tractable, is given in Appendix~\ref{alternative}, and the 
standard structural motifs of ``alpha helices'' and ``beta strands'' are discussed in this alternative model.

Robust invariants of fatgraphs are defined and studied in Section~\ref{protinvariants} providing countless meaningful new protein descriptors.  
Section~\ref{globules} gives the injectivity results
mentioned above after first discussing certain practical aspects of implementing our methods.   Finally, Section~\ref{closing} contains closing remarks 
including several further
biophysical applications of our methods which will appear in companions and sequels to this paper.

\section{Polypeptides}\label{polypeptides}

There are 20 {\it amino acids}\footnote {Strictly speaking, these 20 molecules are the ``standard gene-encoded'' amino acids, i.e., those amino acids determined 
from RNA via the genetic code; in fact, there are a few other non-standard gene-encoded amino acids which are relatively rare in nature and which we shall
ignore here.}, 19 of which have the similar basic chemical structure illustrated in Figure~\ref{fig:amino}a, where $H,C,N,O$ respectively denote hydrogen,
carbon, nitrogen, oxygen atoms, and the {\it residue} $R$ is one of 19 specific possible sub-molecules; the one further amino acid called {Proline} has the 
related chemical structure containing a ring $CCCCN$ of atoms illustrated in Figure~\ref{fig:amino}b.  The residue ranges from a single hydrogen atom for the 
amino acid called Glycine to a sub-molecule comprised of 19 atoms for the amino acid called Tryptophan.  All 20 amino acids
are composed exclusively of $H,C,N,O$ atoms, except for the amino acids called Cysteine and Metionine each of which also contains a single sulfur atom.

\begin{figure}[!h]
\begin{center}
\epsffile{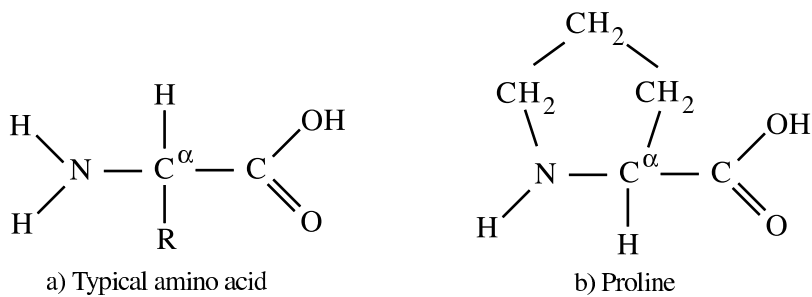}
\caption{Chemical structure of amino acids}
\label{fig:amino}
\end{center}
\end{figure}

In either case of Figure~\ref{fig:amino}, the sub-molecule $COOH$ depicted on the right-hand side is called the {\it carboxyl group}, and the $NH_2$ depicted on 
the left-hand side in Figure~\ref{fig:amino}a
or the $NHC$ on the left-hand side in Figure~\ref{fig:amino}b is called the {\it amine group}.  The carbon atom bonded to the carboxyl and amine groups is 
called the {\it alpha carbon atom} of the amino acid, and it is typically denoted $C^\alpha$.  The alpha carbon atom is bonded to exactly one further atom in 
the residue, either a hydrogen atom in Glycine or a carbon atom, called the {\it beta carbon atom}, in all other cases.

As illustrated in Figure~\ref{fig:polypeptide}, a sequence of $L$ amino acids can combine to form a {\it polypeptide}, where the carbon atom from the carboxyl 
group of $i$th amino acid forms a {\it peptide bond} with the nitrogen atom from the amine group of the $(i+1)$st amino acid  together with the resulting 
condensation of a water molecule comprised of an $OH$ from the carboxyl group of the former and an $H$ from the amine group of the latter, for $i=1, 2,\ldots 
,L-1$.  The nature of this peptide bond and the accuracy of the implied geometry of Figure~\ref{fig:polypeptide} will be discussed presently,
and the further notation in the figure will be explained later.

\begin{figure}[!h]
\begin{center}
\epsffile{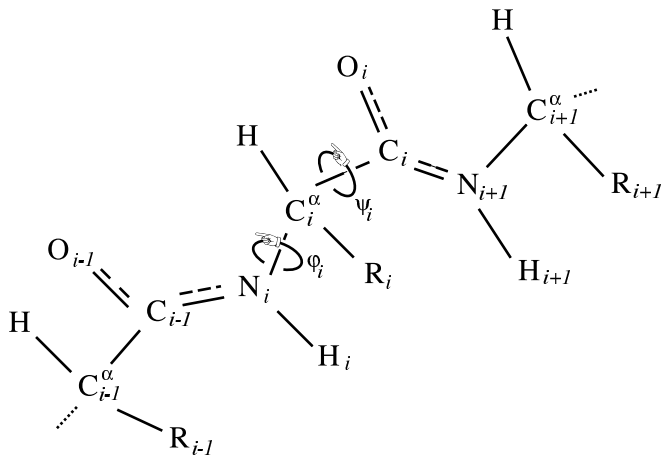}
\caption{A polypeptide}
\label{fig:polypeptide}
\end{center}
\end{figure}

The {\it primary structure} of a polypeptide is the ordered sequence $R_1, R_2,\ldots , R_L$ of residues or of amino acids occurring in this chain, i.e., a word 
in the 20-letter alphabet of amino acids of length $L$, which ranges in practice from $L=3$ to $L\approx 30,000$.    The carbon and nitrogen atoms which 
participate in the peptide bonds together with the alpha carbon atoms form the {\it backbone} of the polypeptide, which is described by 
$$N_1-C_1^\alpha-C_1-N_2-C_2^\alpha-C_2-\cdots -N_i-C_i^\alpha-C_i-\cdots-N_L-C_L^\alpha-C_L,$$ indicating the standard enumeration of atoms along the backbone.  
The first amine nitrogen atom and the the last carboxyl carbon atom, respectively, are called the $N$ and $C$ {\it termini} of the polypeptide.

The $i$th {\it peptide unit}, for $i=1,2,\ldots ,L-1$, is comprised of the consecutively bonded atoms $C_i^\alpha-C_i-N_{i+1}-C_{i+1}^\alpha$ in the backbone 
together with the oxygen atom $O_i$ from the carboxyl group bonded to $C_i$ and one further atom, namely, the remaining hydrogen atom $H_{i+1}$ of the amine 
group except for Proline, for which the further atom is the carbon preceding the nitrogen of the amine group in the Proline ring.

This describes the basic chemical structure of a polypeptide, where the further physico-chemical details about residues, for example, can be found in any 
standard text and will not concern us here.

There are several key geometrical facts about polypeptides as follows, where we refer to the center of mass of the Bohr model of a nucleus as the ``center'' of 
the atom and to the  line segment connecting the centers of two chemically bonded atoms as the ``bond axis''.

\begin{fact}\label{geometricfacts} For any polypeptide, there are the following geometric
constraints:

\smallskip

\noindent{\bf Fact A}~each peptide unit is planar, i.e., the centers of the six constituent atoms of the peptide unit lie in a plane, and furthermore, the 
angles between the bond axes in a peptide unit are always fixed at 120 degrees;

\smallskip

\noindent{\bf Fact B}~at each alpha carbon atom $C^\alpha_i$, the four bond axes (to hydrogen, $C_i, N_i$ and to the residue, i.e., to the hydrogen atom of Glycine or 
to the beta carbon atom in all other cases) are tetrahedral\footnote{Another geometric constraint on any gene-encoded protein is that when viewed along the bond axis from hydrogen to $C_i^\alpha$, the bond axes 
occur in the cycle ordering corresponding to $C_i$, residue, $N_i$.  This imposes various chiral constraints on proteins but plays no role in our basic fatgraph
model.};

\smallskip

\noindent{\bf Fact C}~in the plane of each peptide unit, the centers of the two alpha carbons occur on opposite sides of the line determined by the bond axis of 
the peptide bond, except occasionally for the peptide unit preceding Proline.

\end{fact}

We must remark immediately that these geometric facts are only effectively true, that is, the peptide unit is {\sl almost} planar and the angles 
between bond axes in a peptide unit are {\sl nearly} 120 degrees for example in Fact A; thus, the depiction in Figure~\ref{fig:polypeptide} of the peptide unit 
is nearly geometrically accurate.  In nature, thermal and other fluctuations do slightly affect the geometric absolutes stated in Fact~\ref{geometricfacts}, but 
we shall nevertheless take these facts as geometric absolutes in constructing our model.

Fact A is fundamental to our constructions, and it arises from purely quantum effects: the planar character is provided by the ``sp$^2$ hybridization'' of 
electrons in the $C_i$ and $N_{i+1}$ atoms in the $i$th peptide unit, and the peptide unit is rigid because of additional bonding with $O_i$ of the two
p-electrons from $C_i,N_{i+1}$ not involved in the sp$^2$ hybridization.  This complexity of shared electrons is why the peptide bond and the bond between $C_i$ 
and $O_i$ are often drawn as ``partial double bonds'' as in Figure~\ref{fig:polypeptide}.  In contrast, Fact B is a standard consequence of the valence of 
carbon atoms in the Bohr model absent any quantum mechanical hybridization of electrons.

As a point of terminology, Fact C expresses that except for Proline, the peptide unit occurs in what is called the ``trans-conformation'', and the complementary 
possibility (with the centers of the alpha carbon atoms in a peptide unit on the same side of the line determined by the axis of the peptide bond) is called the 
``cis-conformation''.  This geometric constraint follows from the simple fact that in the cis-conformation, the two ``large'' alpha carbon atoms in the peptide 
unit would be so close together as to be energetically unfavorable.  In contrast for cis-Proline, the two conformations are comparable since in either case, two 
carbons (either the two alpha carbons or one alpha and the beta carbon in the Proline ring) must be close together; nevertheless, cis-Proline, as opposed to 
trans-Proline, occurs only about ten percent of the time in nature since the latter is still somewhat energetically favorable.  This exemplifies a general 
trend: somewhat energetically unfavorable conformations do occur but more rarely than favorable ones, and very energetically unfavorable conformations do not 
occur at all.

The mechanism underlying Fact C is that atoms cannot ``bump into each other'', or more precisely, their centers cannot be closer than their van der Waals radii 
allow, and this is called a  {\it steric constraint},
which will be pertinent to subsequent discussions.

Facts A and B together indicate the basic geometric structure of a polypeptide: a sequence of planar peptide units meeting at tetrahedral angles at the alpha 
carbon atoms; these planes can rotate rather freely about the axes of these tetrahedral bond axes, and this accounts for the relative flexibility of 
polypeptides.
For a polypeptide at equilibrium in some environment, the dihedral angle along the bond axis of $N_i-C_i^\alpha$  (and $C_i^\alpha-C_i$) between the bond axis 
of $C_{i-1}=N_i$ (and $N_i-C_i^\alpha$) and the bond axis of $C_i^\alpha-C_i$ (and $C_i=N_{i+1}$) is called the {\it conformational angle} $\varphi _i$ (and 
$\psi _i$ respectively); see Figure~\ref{fig:polypeptide}.  Illustrating the physically possible pairs $(\varphi _i,\psi _i)\in S^1\times S^1$, steric 
constraints for each amino acid can be plotted in what is called a Ramachandran plot; in particular, for any polypeptide at equilibrium in any environment, 
$\varphi_i$ is bounded away from zero because of steric constraints involving $C_{i-1}$ and $C_i$.

This completes our discussion of the intrinsic physico-chemical and geometric aspects of polypeptides underlying our model.  The remaining such aspect of 
importance to us depends critically upon the ambient environment in which the polypeptide occurs.

An {\it electronegative} atom is one that tends to attract electrons, and examples of such atoms include $C,N,O$ in this order of increasing such tendency.  
When an electronegative atom approaches another electronegative atom which is chemically bonded to a hydrogen atom, the two electronegative atoms can share the 
electron envelope of the hydrogen atom and attract one another through a {\it hydrogen bond}.  A hydrogen bond has a well-defined potential energy
determined on the basis of electrostatics which can be computed\footnote{For instance in the standard method called DSSP \cite{DSSP} where $r_{XY}$ denotes the 
distance between the centers of atoms
$X,Y\in\{ H,N,O\}$ in Angstroms and the location of $H$ is determined from idealized geometry and bond lengths in practice, the assignment of potential energy
to the hydrogen bond  between $O$ and $NH$  in a water environment is given by $q_1q_2\{ r_{ON}^{-1}+r_{CH}^{-1}-r_{OH}^{-1}-r_{CN}^{-1}\} \times 332 ~{\rm
kcal/mole}$, where $q_1=0.42$ and $q_2=0.20$ based on the respective assignment of partial charges $-0.42$e and $+0.20$e
to the carboxyl carbon and amine nitrogen with e representing the election charge.} from the spatial locations of its constituent atoms and the physical
properties of its environment.

For example, the $O_i$ or $N_{i+1}-H_{i+1}$ in one peptide unit can form a hydrogen bond with
the $N_{j+1}-H_{j+1}$ or $O_j$ in another peptide unit, respectively, where $i\neq j$ owing to rigidity and fixed lengths of 1.3-1.6 Angstroms of bond axes.  
For another example, many of the remarkable properties of water arise from the occurrence of hydrogen bonds among $HOH$ and $OH_2$ molecules.  The absolute 
potential energy of hydrogen bonds is rather large, so a polypeptide in a given environment seeks to saturate as many hydrogen bonds as possible subject to 
steric and other
physico-chemical and geometric constraints.  For example in an aqueous environment, the oxygen and nitrogen atoms in the peptide units of a polypeptide might 
form hydrogen bonds with one another or with the ambient water molecules of their environment, and there may also occur hydrogen bonding involving atoms 
comprising the residues or the alpha carbons.

Suppose that a polypeptide is at equilibrium, i.e., at rest, in some environment.  Its {\it tertiary structure} in that environment is the specification of the 
spatial coordinates of the centers of all of its constituent atoms.  Furthermore, fix some {\it energy cutoff} and regard a pair  $O_i$ and $N_{j}$ of backbone 
atoms as
being hydrogen bonded if the potential energy discussed above is less than this energy cutoff; a standard convention is to take the energy cutoff to be -0.5 
kcal.mole\footnote{Other methods
\cite{Kortemmeetal,Lindaueretal} of determining hydrogen bonds are also employed.}.  The {\it secondary structure} of the polypeptide\footnote{This is a slight
abuse of terminology as biologists might call this rather ``super secondary structure''; we shall explain this distinction further when it is appropriate.} at
equilibrium in an environment is the specification of hydrogen bonding as determined by an energy cutoff among its constituent backbone atoms $O_i$ and $N_j$, 
for $i,j=1,2,\ldots ,L$.

Certain polypeptides occur as the ``proteins'' which regulate and effectively define life as we know it.
The collective knowledge of protein primary structures is deposited in the manually curated SWISS-PROT data bank \cite{Swiss-Prot}, which contains about 400,000 
distinct entries, and the computer curated UNI-PROT data bank \cite{Uni-Prot}, which contains about 6,000,000 entries.   These data are readily accessible at 
www.ebi.ac.uk/swissprot and www.uniprot.org respectively.  The collective knowledge of protein tertiary structure is deposited in the Protein Data Bank (PBD) 
\cite{PDB}, which contains roughly 55,000 proteins at this moment, where the atomic locations of each of the constituent atoms of each of these proteins is 
recorded; each entry in the PDB, i.e., each protein, thus comprises a vast amount of data.  Atomic locations in PDB should be taken with an experimental 
uncertainty of 0.2 Angstroms, and  the conformational angles $\varphi,\psi$ computed from them should be taken with an experimental uncertainty of 15-20 
degrees; however, the unit displacement vectors of bond axes along the backbone, upon which our model is based, are substantially better determined  \cite{Finkelstein-personal}.

Upon postulating definitions of the various secondary structure elements in terms of properties of the atomic locations, protein secondary structure can be 
calculated from tertiary structure.  A standard such method is called  the Dictionary of Secondary Structures for Proteins (DSSP) \cite{DSSP}, and proprietary 
software for these calculations and DSSP files for each PDB entry can be found at swift.cmbi.ru.nl/gv/dssp.  Hydrogen bond strengths and various conformational 
angles are also output as part of the calculations of DSSP.

~~\vfill\eject

\section{Fatgraphs}\label{fatgraphs}

\vskip .2in

\subsection{Surfaces}\label{surfaces}

According to the classification of surfaces \cite{Massey}, a compact and connected surface $F$ is uniquely determined up to homeomorphism by the specification 
of whether it is orientable together with its genus $g=g(F)$ and number $r=r(F)$ of boundary components, or equivalently, by either $g$ or $r$ and its Euler 
characteristic $$\chi = \chi (F)=\begin{cases} 2-2g-r,& {\rm if}~{\it F}~{\rm  is ~orientable};\cr 2-g-r,& {\rm if}~{\it F}~{\rm is~non-orientable}.\end{cases}$$
It is useful to define the {\it modified genus} of a connected surface $F$ to be
$$g^*=g^*(F)=\begin{cases}
g,&~{\rm if}~{\it F}~{\rm is~orientable};\\
g/2,&~{\rm if}~{\it F}~{\rm is~non-orientable},\\
\end{cases}$$
so the formula $\chi=2-2g^*-r$ holds in either case.

Recall \cite{Massey} that the {\it orientation double cover} of a surface $F$ is the oriented surface $\tilde F$ together with the continuous map $p:\tilde F\to 
F$
so that for every point $x\in F$ there is a disk neighborhood $U$ of $x$ in $F$, where $p^{-1}(U)$ consists of two components on each of which
$p$ restricts to a homeomorphism and
where the further restrictions of $p$ to the boundary circles of these two components
give both possible orientations of the boundary circle of $U$.
Such a covering
$p:\tilde F\to F$ always exists, and its properties uniquely determine $\tilde F$ up to homeomorphism and $p$ up to its natural equivalence.  In particular, if 
$F$ is connected and orientable, then $\tilde F$ has two components with opposite orientations, each of which is identified with $F$ by $p$.
Furthermore provided $F$ is connected, $F$ is non-orientable if and only if $\tilde F$ is connected, and a closed curve in $F$ lifts to a closed curve in 
$\tilde F$ if and only if a neighborhood of it in $F$ is homeomorphic to an annulus.

\subsection{Fatgraphs and their associated surfaces}\label{fatgraphsequiv}

Consider a finite graph $G$ in the usual sense of the term comprised of vertices $V=V(G)$ and edges $E=E(G)$, which do not contain their endpoints and where an 
edge is not necessarily uniquely determined by its endpoints, or in other words, $G$ is a finite one-dimensional CW complex.  Our standard notation will be 
$v=v(G)=\# V$ and $e=e(G)=\# E$, where $\# X$ denotes the cardinality of a set $X$.  To avoid cumbersome cases in what follows, we shall assume that no 
component of $G$ consists of a single vertex or a single edge with distinct endpoints.  Removing a single point from each edge produces a
subspace of $G$, each component of which is called a {\it half-edge}.  A half-edge which contains $u\in V$ in its closure is said to be {\it incident} on $u$, 
and the number of distinct half-edges incident on $u$ is the {\it valence} of $u$.

A {\it fattening} on $G$ is the specification of a cyclic ordering on the half-edges incident on $u$ for each $u\in V$, and an
{\it $X$-coloring} on $G$ is a function $E\to X$, for any set $X$.

A {\it fatgraph} $G$ is a graph endowed with a fattening together with a coloring by a set with two elements, where we shall refer to the two colors on edges as 
``twisted'' and ``untwisted''.  A fatgraph $G$ uniquely determines a surface $F(G)$ with boundary as follows.

\begin{figure}[!h]
\begin{center}
\epsffile{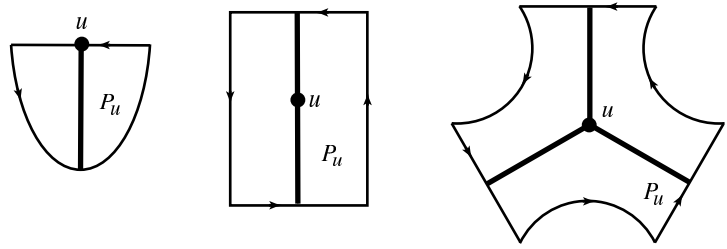}
\caption{The polygon $P_u$ associated with a vertex $u$}
\label{fig:polygonforfatgraph}
\end{center}
\end{figure}

\begin{construction}\label{constructsurface}
For each vertex $u\in V$ of $G$ of valence $k\geq 2$, consider an oriented surface diffeomorphic to a polygon $P_u$ of $2k$ sides containing in its interior a 
single vertex of valence $k$ each of whose incident edges are also incident on a univalent vertex contained in alternating sides of $P_u$, which are identified 
with the half-edges of $G$ incident on $u$
so that the induced counter-clockwise cyclic ordering on the boundary of $P_u$ agrees with the fattening of $G$ about $u$; for a vertex $u$ of valence $k=1$, 
the corresponding surface $P_u$ contains $u$ in its boundary; see Figure~\ref{fig:polygonforfatgraph}.
The surface $F(G)$ is the quotient of the disjoint union $\sqcup_{u\in V} P_u$, where the frontier edges,
which are oriented with the polygons on their left, are identified by a homeomorphism if the corresponding half-edges lie in a common edge of $G$, and this 
identification of oriented segments is orientation-preserving if and only if the edge is twisted.   The graphs in the polygons $P_u$, for $u\in V$, combine to 
give
a fatgraph embedded in $F(G)$ with its univalent vertices in the boundary, which is identified with $G$ in the natural way so that we regard $G\subseteq F(G)$.
\end{construction}

Our standard notation will be to set
$$\aligned
r(G)&=r(F(G))=~{\rm the~number~of~boundary~components~of}~F(G),\\
g^*(G)&=g^*(F(G))=~{\rm the~modified~genus~of}~F(G).\\
\endaligned$$

It is often convenient to regard a fatgraph more pictorially by considering the planar projection of a graph embedded in three-space, where the cyclic ordering 
is given near each vertex by the counter-clockwise ordering in the plane of projection and edges can be drawn with arbitrary under/over crossings; we also 
depict untwisted edges as ordinary edges and indicate twisted edges with an icon $\times$, or more generally, take this as defined modulo two so that an even 
number of icons $\times$ represents an untwisted edge and an odd number represents a twisted edge. Several examples of fatgraphs and their corresponding 
surfaces are illustrated in Figure~\ref{fig:skinnysurfaces}, where the bold lines indicate the planar projection of the fatgraph, the dotted lines indicate the 
gluing along edges of polygons, and the further notation in the figure will be explained later.

\begin{figure}[!h]
\begin{center}
\epsffile{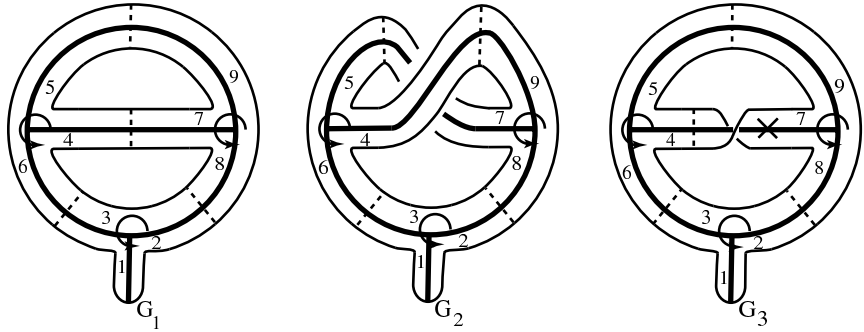}
\caption{The surface associated to a fatgraph}
\label{fig:skinnysurfaces}
\end{center}
\end{figure}

The graph $G$ is evidently a strong deformation retract of  $F(G)$, so the Euler characteristic is $\chi(F(G))=\chi(G)=v(G)-e(G)$, and the boundary components of $F(G)$ 
are composed of the frontier edges of $\sqcup_{u\in V} P_u$ which do not correspond to half-edges of $G$.

\begin{proposition}\label{changetwists}
Suppose that $G$ is a fatgraph and $X,Y\subseteq E(G)$ are disjoint collections of edges.
Change the color, twisted or untwisted, of the edges in $X$ and delete from $G$
the edges in $Y$ to produce another fatgraph $G'$, whose cyclic orderings on half-edges are induced from those on $G$ in the natural way.
Then $|r(G)-r(G')|~\leq ~\#X+\#Y$.
\end{proposition}

\begin{proof}
By the triangle inequality, it suffices to treat the case that $X\cup Y=\{ f\}$, and we set $r=r(G)$. If $f\in E(G)$ is incident on a univalent vertex, then 
neither changing the color of nor deleting $f$ alters $r$, so we may assume that this is not the case.  Consider an arc $a$ properly embedded in $F(G)$ meeting 
$f$ in a single transverse
intersection and otherwise disjoint from $G$.   Rather than changing the color on $f$ to produce $G'$, let us instead cut $F(G)$ along $a$ and then re-glue 
along the two resulting copies
of $a$ reversing orientation to produce a surface homeomorphic to $F(G')$.  If the endpoints of $a$
occur in a common boundary component of $F(G)$, then the change of color on
$f$ either leaves $r$ invariant or increases it by one, and if they occur in different boundary components, then the change of color on $f$ necessarily 
decreases
$r$ by one.  For the remaining case, rather than removing the edge $f$ to produce $G'$, let us instead consider cutting $F(G)$ along $a$ to produce a surface 
homeomorphic to $F(G')$.  If the endpoints of $a$ occur in the same boundary component of $F(G)$, then cutting on $a$ either leaves $r$ invariant or increases 
it by one, and if they occur in different boundary components, then the cut on $a$ decreases $r$ by one.
\end{proof}

We say that a fatgraph $G$ is {\it untwisted} if all of its edges are untwisted, and  this is evidently a sufficient but not a necessary condition for $F(G)$ to 
be orientable.

\begin{remark} Suppose that $G$ is an untwisted fatgraph.  Let us emphasize that the genus of $F(G)$ is {\sl not} the classical genus of the underlying graph, 
i.e., the least genus orientable surface in which the underlying graph can be embedded.  Rather, the classical genus of the underlying graph is the least genus 
of an orientable surface $F(G)$ arising from all possible fattenings on the underlying graph.
\end{remark}

We say that two fatgraphs $G_1$ and $G_2$ are {\it strongly equivalent} if there is an isomorphism of the graphs underlying $G_1$ and $G_2$ that respects the 
cyclic orderings and preserves the coloring and that they are {\it equivalent} if there is a homeomorphism from $F(G_1)$ to $F(G_2)$ that maps $G_1\subseteq 
F(G_1)$ to $G_2\subseteq F(G_2)$.   It is clear that strong equivalence implies equivalence and that equivalence implies that the corresponding surfaces are
homeomorphic; neither converse holds in general.

Given a vertex $u$ of $G$, define the {\it vertex flip} of $G$ at $u$ by reversing the cyclic ordering on the half-edges incident on $u$ and adding another icon 
$\times$ to each half-edge incident on $u$.   In particular, a vertex flip on a univalent vertex simply adds an icon $\times$ to the edge incident upon it.

\begin{proposition}\label{strong}  Two untwisted fatgraphs are equivalent if and only if they are strongly equivalent.
Two arbitrary fatgraphs $G_1$ and $G_2$ are equivalent if and only if there is a third fatgraph $G$ which arises from $G_1$ by a finite sequence of vertex flips 
so that $G$ and $G_2$ are strongly equivalent.   In particular, if $G$ arises from $G_1$ by a vertex flip, then
$G$ and $G_1$ are equivalent.
\end{proposition}

\begin{proof}
In case $G_1$ and $G_2$ are untwisted, a homeomorphism from $F(G_1)$ to $F(G_2)$ mapping $G_1$ to $G_2$ restricts to a strong equivalence of $G_1$ and $G_2$, 
and the converse follows by construction in any case, as already observed, thus proving the first assertion.

The third assertion follows since a flip on a vertex $u$ of $G_1$ corresponds to simply
reversing the orientation of the polygon $P_u$ in the construction of $F(G)$, i.e., in our graphical
depiction, removing the neighborhood of $u$ from the plane of projection, turning it upside down in three-space, and then replacing it in the plane of 
projection at the expense of twisting one further time each incident half-edge of $G$;
this evidently extends to a homeomorphism of $F(G_1)$ to $F(G)$ which maps $G_1$ to $G$ but does not preserve coloring.

Since strong equivalence implies
equivalence by construction and equivalence of fatgraphs is clearly a transitive relation, if there is such a fatgraph $G$ as in the statement of the 
proposition, then $G_1$ and $G_2$ are indeed equivalent.  For the converse, we may and shall assume that $G_1$ and $G_2$ are connected.

Consider a fatgraph $G$ with $v$ vertices and $e$ edges, and choose a maximal tree $T$ of $G$.  There are  $1-\chi(G)=1-v+e$ edges in $G-T$ since we may 
collapse $T$ to a point without changing $v-e$, which is therefore the Euler characteristic of the collapsed graph comprised of a single vertex and one edge for 
each edge of $G-T$.

We claim that there is a composition of flips of vertices in $G$ that results in a fatgraph with any specified twisting on the edges in $T$.  To see this, 
consider the collection of all functions from the set of edges of $G$ to ${\mathbb Z}/2$, a set with cardinality $2^e$.  Vertex flips act on this set of 
functions in the natural way, and there are evidently $2^v$ possible compositions of vertex flips.  The simultaneous flip of all vertices of $G$ acts trivially 
on this set of functions and corresponds to reversing the cyclic orderings at all vertices, so only $2^{v-1}$ such compositions may act non-trivially.  Insofar 
as $2^e/2^{v-1}=2^{1-v+e}$ and there are $1-v+e$ edges of $G-T$ by the previous paragraph, the claim follows.

Finally, suppose that $G_1$ and $G_2 $ are equivalent and let $\phi:F(G_1)\to F(G_2)$ be a homeomorphism of surfaces that restricts to a homeomorphism of $G_1$ 
to $G_2$.
Performing a vertex flip on $G_1$ and identifying edges before and after in the natural way produces a fatgraph in which $T$ is still a maximal tree and
which is again equivalent to $G_2$, according to previous remarks, by a homeomorphism still denoted
$\phi$, which maps $T$ to the maximal tree $\phi(T)\subset G_2$.
By the previous paragraph, we may apply a composition of vertex flips to $G_1$ to produce a fatgraph $G$ so that an edge of the maximal tree $T\subset G$ is 
twisted if and only if its image under $\phi$
is twisted.

Adding an edge of $G-T$ to $T$ produces a unique cycle in $G$, and a neighborhood of this cycle in $F(G)$ is either an annulus or a M\"obius band with a similar 
remark for edges of $G_2-\phi(T)$.  Since $\phi$ restricts to a
homeomorphism of the corresponding annuli or M\"obius bands in $F(G)$ and $F(G_2)$, an
edge of $G-T$ is twisted if and only if its image under $\phi$ is twisted.  It follows that $G$ and $G_2$ are strongly equivalent as desired.
\end{proof}

\subsection{Fatgraphs and permutations}\label{fatgraphsperms}

We shall adopt the standard notation for a permutation on a set $S$ writing $(s_1,\ldots ,s_k)$ for the cyclic permutation $s_1\mapsto s_2\mapsto\cdots\mapsto 
s_k\mapsto s_1$ on distinct elements $s_1,\ldots , s_k\in S$, called a {\it transposition} if $k=2$, and shall
compose permutations $\sigma,\tau$ on $S$ from right to left, so that $\sigma\circ\tau(s)=\sigma(\tau(s))$.  An {\it involution} is a permutation $\tau$ so that 
$\tau\circ\tau={1_S}$, where $1_S$ denotes the identity map on $S$.  Two permutations are {\it disjoint} if they have disjoint supports, so disjoint 
permutations necessarily commute.

Fix a fatgraph $G$.  A {\it stub} of $G$ is a half-edge which is not incident on a univalent vertex of $G$. There are exactly two non-empty connected fatgraphs 
with no stubs, namely, the two we have proscribed consisting of a single vertex with no incident half-edges and a single edge with distinct endpoints.

A fatgraph $G$ determines a triple $(\sigma (G),\tau_u(G),\tau_t(G)) $ of permutations on its set $S=S(G)$ of stubs as follows.

\begin{construction}\label{constructperms}
For each vertex $u$ of $G$ of valence $k\geq 2$ with incident stubs $s_1,\ldots ,s_{k(u)}$ in a linear ordering compatible with the cyclic ordering given by the 
fattening on $G$, consider the cyclic permutation $(s_1,\ldots ,s_{k(u)})$.  By construction, the cyclic permutations corresponding to distinct vertices of $G$ 
are disjoint.  The composition $$\sigma (G)=\prod_{{\{ {\rm vertices}~u\in V:}\atop{~u~{\rm has~valence}\geq 2\}}}(s_1,\ldots ,s_{k(u)})$$ is thus well-defined 
independent of the order in which the product is taken and likewise for the compositions of transpositions
$$\aligned
\tau _u(G)&=\prod _{{\{
{\rm pairs~of~distinct~stubs}~h,h'~{\rm contained}\atop{~{\rm in~some~untwisted~edge~of}~G
 \}}}}
  (h,h'),\cr
\tau _t(G)&=\prod _{{\{
{\rm pairs~of~distinct~stubs}~h,h'~{\rm  contained}\atop{~{\rm in~some~twisted~edge~of}~G}
\}}}
(h,h').\cr
\endaligned$$
\end{construction}

Notice that $\sigma(G)$ has no fixed points because we have taken the product over vertices of valence at least two,
and $\tau_u(G)$ and $\tau_t(G)$ are disjoint involutions whose fixed points are the stubs corresponding to the univalent vertices of $G$.

For example, enumerating the stubs of the fatgraphs $G_1,G_2,G_3$ as illustrated in Figure~\ref{fig:skinnysurfaces}, we have:
$$\aligned
\sigma(G_1)&=\sigma(G_2)=\sigma(G_3)=(1,2,3)(4,5,6)(7,8,9),\cr
\tau_u(G_1)&=(2,8)(3,6)(4,7)(5,9), \tau_t(G_1)={1_S},\cr
\tau_u(G_2)&=(2,8)(3,6)(4,9)(5,7), \tau_t(G_2)={1_S},\cr
\tau_u(G_3)&=(2,8)(3,6)(5,9), \tau_t(G_3)=(4,7).\cr
\endaligned$$

\begin{remark}
There is another treatment of fatgraphs as triples of permutations on the set of all half-edges instead of stubs, where the univalent vertices are expressed as 
fixed points of the analogue of $\sigma$.   Moreover, there is a transposition in the analogue of $\tau_u\circ\tau_t$ corresponding to each half-edge, but the 
formulation we have given here,
which treats univalent vertices as ``endpoints of half-edges rather than endpoints of edges'',  does not require these additional transpositions.  Since our 
model will have a plethora of univalent vertices, we prefer
the more ``efficient'' version described above, which is just a notational convention for permutations.
\end{remark}

Define a {\it labeling} on a fatgraph $G$ with $N$ stubs to be a linear ordering on its stubs, i.e., a bijection from the set of stubs of $G$ to the set $\{ 
1,2,\ldots ,N\}$.

\begin{proposition}\label{strongclasses} Fix some natural number $N\geq 2$.
The map $G\mapsto (\sigma(G),\tau_u(G),\tau_t(G))$ of Construction~\ref{constructperms}
induces a bijection between
the set of strong equivalence classes of fatgraphs with $N$ stubs and the set of all conjugacy classes
of triples $(\sigma, \tau_u,\tau_t)$ of permutations on $N$ letters, where $\sigma$ is fixed point free and $\tau_u$ and $\tau_t$
are disjoint involutions.
\end{proposition}

\begin{proof}
The assignment $G\mapsto (\sigma(G),\tau_u(G),\tau_t(G))$ induces a mapping from the set of labeled fatgraphs with $N$ stubs to the set of triples of 
permutations on $\{ 1,2,\ldots ,N\}$ in the natural way.  This induced mapping has an obvious two-sided inverse, where the labeled fatgraph is constructed 
directly from the triple of permutations; we are here using our convention that no component of $G$ is a single vertex or a single edge with distinct univalent 
endpoints.  A strong equivalence of labeled fatgraphs induces a bijection of $ \{ 1,2,\ldots ,N\}$ which conjugates their corresponding triples of permutations 
to one another and conversely, so the result follows.
\end{proof}

\begin{construction}\label{constructdoublecover} Suppose that $G$ is a fatgraph with triple
$(\sigma,\tau_u,\tau_t)$ of permutations on its set $S$ of stubs determined by Construction~\ref{constructperms}.
Construct a new set $\bar S=\{ \bar s:s\in S\}$ and a new permutation $\bar\sigma$ on $\bar S$ where there is
one $k$-cycle $(\bar s_k,\ldots ,\bar s_1)$ in $\bar\sigma$ for each $k$-cycle $(s_1,\ldots ,s_k)$ in $\sigma$.
Construct from $\tau_u$ a new permutation $\bar\tau_u$ on $\bar S$, where there is one transposition $(\bar s_1,\bar s_2)$ in $\bar\tau_u$ for each 
transposition $(s_1,s_2)$ in $\tau _u$, and construct yet another new permutation $\bar\tau_t$ on $S\sqcup \bar S$ from $\tau_t$, where there are two 
transpositions
$(\bar s_1,s_2)$ and $(s_1,\bar s_2)$ in $\tilde\tau _t$ for each transposition $(s_1,s_2)$ in $\tau_t$.
Finally, define
permutations on $S\sqcup \bar S$ by
$$\aligned
\sigma'&=\sigma\circ\bar\sigma,\\
\tau'&=\tau_u\circ\bar\tau_u\circ \bar\tau_t,\\
\endaligned$$
where the order of composition on the right-hand side is immaterial because
the permutations are disjoint in each case.\end{construction}

\begin{proposition}\label{doublecover}  Suppose that Construction~\ref{constructperms} assigns the triple
$(\sigma ,\tau_u,\tau_t)$ of permutations to the fatgraph $G$ with set $S$ of stubs, let $\sigma ',\tau '$ be determined
from them according to Construction~\ref{constructdoublecover}, and consider the untwisted fatgraph $G'$
determined by Construction~\ref{constructperms} from the triple $(\sigma ',\tau ',1_{S\sqcup\bar S})$.  Then $F(G')$ is the orientation double cover of $F(G)$, and the covering transformation is 
described by $s\leftrightarrow \bar s$.  
In particular provided $F(G)$ is connected, $F(G')$ is connected if and only if
$F(G)$ is non-orientable.  Furthermore, there is a one-to-one correspondence between the boundary components
of $F(G')$ and the orientations on the boundary components of $F(G)$, i.e., $F(G')$ has twice as many
boundary components as $F(G)$.
\end{proposition}

\begin{proof}
The surface $F(G')$ has the required properties of the orientation double cover by construction, so the first two claims follow from the general principles 
articulated in Section~\ref{surfaces}.  Since each boundary component of $F(G)$ evidently has a neighborhood in $F(G)$ homeomorphic to an annulus, the final 
assertion follows as well.
\end{proof}

\begin{proposition} \label{boundaryandconn}  Adopt the hypotheses and notation of Proposition~\ref{doublecover}
and consider the composition $\rho'=\sigma'\circ\tau'$.
\smallskip

\noindent {\rm ~i)}~ The orientations on the boundary components of $F(G)$ are in one-to-one correspondence with the cycles of  $\rho'$.
More explicitly, suppose that $s_1^1s_1^2s_2^1s_2^2\cdots s_n^1s_n^2$ is the ordered sequence of stubs traversed
by an oriented edge-path in $G$ representing a boundary component of $F(G)$ with some orientation, where $s_j^1,s_j^2$ are contained in a common edge of $G$ and 
perhaps $s_j^1=s_j^2$ if they are contained in an edge incident on a univalent vertex, for $j=1,2,\ldots ,n$.  Erasing the bars on elements from the 
corresponding cycle of $\rho '$ produces the sequence
$(s_1^2,s_2^2,\ldots ,s_n^2)$ of stubs of $G$ serially traversed by the corresponding oriented
boundary component of $F(G)$, called a {\rm reduced cycle} of $\rho'$.

\smallskip

\noindent {\rm ii)}~There is the following algorithm to determine whether $G$ is connected in terms of the associated
triple $(\sigma,\tau_u,\tau_t)$ of permutations.  For any linear ordering on $S$, let
$X$ be the subset of $S$ in the reduced cycle of $\rho '$ containing the first stub.
(*) If $X=S$, then $G$ is connected, and the algorithm terminates.    If $X\neq S$, then consider the existence of a least stub $s\in X-S$ so that $\tau(s)\in 
X$.
If there is no such stub $s$, then $G$ is not connected, and the algorithm terminates. If there is such a stub $s$, then update $X$ by adding to it the subset 
of $S$ in the
reduced cycle of $\rho'$ containing $s$.   Go to (*).
\end{proposition}

\begin{proof} Let us first consider the case that $\tau_t=1_{S\sqcup \bar S}$, i.e., $G$ is untwisted, and set $\tau=\tau_u$.

For the first part, consider a stub $s$ of $G$ and the effect of $\sigma\circ\tau$ on $s$.  The stub $s$ is contained in an edge incident on a univalent vertex 
if and only if $s$ is a fixed point of $\tau$ by construction, and
$\sigma(s)=\sigma(\tau(s))$ in this case is the stub following $s$ in the cyclic ordering at the non-univalent endpoint of this edge.  In the contrary case that 
$s$ is not a fixed point of $\tau$,  the stubs $s$ and $\tau (s)$ are half-edges contained in a common edge of $G$, and $s,\tau(s),\rho(s)=\sigma(\tau(s))$ is 
likewise a consecutive triple of
stubs occurring in an edge-path of $G$ corresponding to a boundary component of $F(G)$ oriented with $F(G)$ on its left.
It follows that a cycle of $\sigma\circ\tau$ is comprised of every other stub traversed by an edge-path in $G$ which corresponds to a boundary component of 
$F(G)$ oriented in this way, proving the first part.

For the second part, the collection of stubs in $X$ always lies in a single component of $G$ in light of the previous remarks, so if at some stage of the 
algorithm $X=S$, then $G$ is indeed connected.  If at some stage of the algorithm there is no stub $s$ with $\tau(s)\in X$, then $X$ is comprised of all the 
stubs in some component of $G$ in light of the previous discussion, so $X\neq S$ in this case implies that $G$ has at least two components.

Turning now to the general case, $F(G')$ is the orientation double cover of $F(G)$,
and the induced projection map on stubs 
just erases the bars by Proposition~\ref{doublecover}.  The proof in this case is therefore entirely analogous.
\end{proof}

To exemplify these constructions and results for the fatgraphs illustrated in Figure~\ref{fig:skinnysurfaces}, we find
$$\aligned
\sigma(G_1)\circ\tau_u(G_1)&=(5,7)(3,4,8)(1,2,9,6),\cr
\sigma(G_2)\circ\tau_u(G_2)&=(1,2,9,5,8,3,4,7,6).\cr
\endaligned$$
Thus, $r(G_1)=3$ and $r(G_2)=1$, and since
$\chi(G_1)=\chi(G_2)=-1$, the (modified) genera are
$g^*(G_1)=0$ and  $g^*(G_2)=1$.

\begin{figure}[!h]
\begin{center}
\epsffile{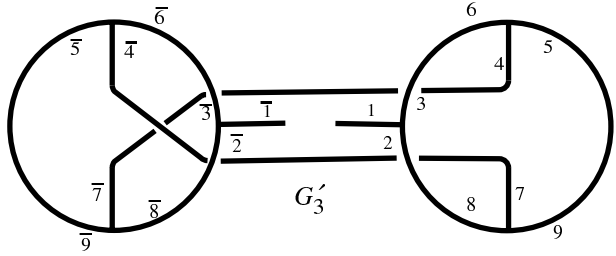}
\caption{Example of the orientation double cover}
\label{fig:doublecover}
\end{center}
\end{figure}

As to $G_3$, according to Construction~\ref{constructdoublecover} and Proposition~\ref{doublecover}, the permutations for the orientation double cover are given 
by
$$\aligned
\sigma ' &=(1,2,3)(4,5,6)(7,8,9)(\bar 3,\bar 2,\bar 1)(\bar 6,\bar 5,\bar 4)(\bar 9,\bar 8,\bar 7),\\
\tau ' &=(2,8)(3,6)(5,9)(\bar 2,\bar 8)(\bar 3,\bar 6)(\bar 5,\bar 9)(4,\bar 7)(\bar 4,7).\\
\endaligned$$
The untwisted fatgraph $G_3'$ corresponding to $(\sigma',\tau',1_{S(G_3)\sqcup \bar S(G_3)})$ is illustrated
in Figure~\ref{fig:doublecover},
and it is connected reflecting the fact that $F(G_3)$ is non-orientable.
The cycles of $\rho'=\sigma'\circ\tau'$ are given by
$$(1,2,9,6), (\bar 1,\bar 3,\bar 5,\bar 8), ~{\rm and}~(\bar 2,\bar 7,5,7,\bar 6),  (3,4,\bar 9,\bar 4,8)$$
corresponding to the oriented boundary cycles of $G_3'$,
and the reduced cycles of $\rho'$ are therefore
$$(1,2,9,6), (1,3,5,8), ~{\rm and}~(2,7,5,7,6),  (3,4,9,4,8),$$
each pair corresponding to the two orientations of a single boundary component of $F(G_3)$.
It follows that $r(G_3)=2$ and thus $g^*(G_3)=1/2$ since again $\chi(G_3)=-1$.

\subsection{Fatgraphs on the computer}\label{fatgraphscomputer}

Given a linear ordering on the vertices of a fatgraph, we may choose an {\sl a priori} labeling on it that is especially convenient, where the stubs about a 
fixed vertex are consecutive and the stubs about each vertex precede those of each succeeding vertex as in Figure~\ref{fig:skinnysurfaces}.
Owing to Proposition~\ref{strongclasses}, the strong equivalence class of a fatgraph $G$  with set $S$ of stubs can conveniently be stored on the computer as a 
triple $(\sigma, \tau_u,\tau_t)$ of permutations on the labels
$\{ 1,2,\ldots ,\#S\}$ of stubs.  The number of non-univalent vertices of $G$ is the number of disjoint cycles in $\sigma$, the number of edges of $G$ which are 
not incident on a univalent vertex is the number of disjoint transpositions in $\tau_u\circ\tau_t$, and the Euler characteristic of $G$ or $F(G)$ is
given by the former minus the latter.  Construction~\ref{constructdoublecover} provides an algorithm, which is easily implemented on the computer, to produce a 
triple $(\sigma ',\tau ', 1_{S\sqcup \bar S})$ from $(\sigma, \tau_u,\tau_t)$ which determines an untwisted fatgraph $G'$ whose corresponding surface $F(G')$ is 
the orientation double cover of $F(G)$ according to Proposition~\ref{doublecover}.  Proposition~\ref{boundaryandconn}i provides an algorithm to determine the 
compatibly oriented boundary components of $F(G')$ and hence the boundary components of $F(G)$ itself, and Proposition~\ref{boundaryandconn}ii then gives an 
algorithm to determine whether $G'$ is connected from this data, where both of these methods are again easily implemented on the computer.

In our applications of these techniques, the fatgraph $G$ will typically be connected as we now assume.  The orientability of $F(G)$ can thus be ascertained 
from the connectivity of $F(G')$.  The boundary components of $F(G)$, and their number $r$ in particular, can be determined, as above, and hence the modified 
genus $g^*=(2-r-\chi)/2$ is likewise easily computed.  Thus, the topological type of $F(G)$ can be determined algorithmically on the computer from the triple 
$(\sigma, \tau_u,\tau_t)$ of permutations for a connected fatgraph $G$, and the particular edge-paths in $G$ corresponding to boundary components of 
$F(G)$ can be ascertained from the cycles of $\sigma'\circ\tau '$.

\section{The model}\label{model}

We take as input to the method the specification for a polypeptide at equilibrium in some environment the following data:

\smallskip

\leftskip .4cm

\noindent {\bf Input i)} the primary structure given as a sequence $R_i$ of letters in the 20-letter alphabet of
amino acids, for $i=1,\ldots, L$;

\smallskip

\noindent {\bf Input ii)} the specification of hydrogen bonding among the various
nitrogen and oxygen atoms $\{ N_i,O_i:i=1,\ldots ,L\}$ described as a collection ${\mathcal B}$ of
pairs $(i,j)$ indicating that $N_i-H_i$ is hydrogen bonded to $O_j$, where $i,j\in\{ 1,\ldots ,L\}$;

\smallskip

\noindent {\bf Input iii)} the displacement vectors $\vec x_i$ from $C_{i}$ to $N_{i+1}$,
$\vec y_i$ from $C_{i}^\alpha$ to $C_i$, and
$\vec z_i$ from $N_{i+1}$ to $C^\alpha_{i+1}$ in each peptide unit, for $i=1,\ldots ,L-1$.

\smallskip

\leftskip=0ex

\noindent These data, which we shall term a {\it polypeptide structure} $P$, are either immediately given in or readily derived from the PDB and DSSP files for 
a folded protein.  Practical and other details concerning the determination of these inputs will be discussed in Section~\ref{implementation}.

A fatgraph is constructed from a polypeptide structure in two basic steps: modeling the backbone using the planarity of the peptide units and the conformational 
geometry along the backbone based on inputs i) and iii), and then adding edges to the model of the backbone for the hydrogen bonds based on inputs ii) and 
iii).

We shall assume that input ii) is consistently based upon fixed energy thresholds with each nitrogen or oxygen atom involved in at most one hydrogen bond (so-called ``simple'' hydrogen bonding)
and relegate the discussion of more general models (with so-called ``bifurcated'' hydrogen bonding)
to Section~\ref{generalmodel}.  The assumption thereby imposed on
${\mathcal B}$ in input ii) is that if $(i,j),(i',j')\in{\mathcal B}$, then $i=i'$ if and only if $j=j'$.

To each peptide unit is associated a fatgraph building block as illustrated in Figure~\ref{fig:buildingblock}.  These building blocks are concatenated to
produce a model of the backbone as illustrated in Figure~\ref{fig:concatenate},
where the determination of whether the edge connecting the two building blocks is twisted is based on input iii).
Specifically,
we shall associate to each peptide unit a positively oriented orthonormal 3-frame determined from input iii).  A pair of consecutive peptide units thus gives a 
pair of
such 3-frames, and there is a unique element of the Lie group $SO(3)$ mapping one to the other.  Using this, we may assign an element of $SO(3)$ to each 
oriented edge of the graph underlying
the fatgraph model and thereby determine an ``$SO(3)$ graph connection'' (cf.\ the next section) on the underlying graph, which
is a fundamental and independently interesting aspect of our constructions.  This assignment is discretized using the bi-invariant metric on $SO(3)$ to 
determine twisting and define the fatgraph model of the backbone, where there are special considerations to handle the case of cis-Proline, which can be 
detected using inputs i) and iii).

Edges are finally added to this model of the backbone in the natural way, one edge for each hydrogen bond in input ii); see Figure~\ref{fig:addhydrogen}.   
These added edges for hydrogen bonds may be twisted or untwisted, and this determination is again made by considering the $SO(3)$ graph connection.

Section~\ref{connections} discusses generalities about 3-frames and $SO(3)$ graph connections.
Section~\ref{backbonesec} details the concatenation of fatgraph building blocks to construct the model of the backbone, and Section~\ref{hydrogen} explains the 
addition of edges corresponding to hydrogen bonds thus completing the description of the basic model.   The final Section~\ref{generalmodel} discusses the 
general model with bifurcated hydrogen bonds  plus other innovations and extensions of the method.
An alternative to the basic model, which gives
an equivalent but not strongly equivalent fatgraph that is arguably more natural than the basic model, is discussed in Appendix~\ref{alternative},
and the standard motifs of polypeptide secondary structure are described in the alternative model.

\subsection{$SO(3)$ graph connections and 3-frames}\label{connections}

The Lie group $SO(3)$ is the group of three-by-three matrices $A$ whose
entries are real numbers satisfying $A A^t=I$ and $\det(A) = 1$,
where $A^t$ denotes the transpose of $A$
and $I$ denotes the identity matrix.
A {metric}
$d: SO(3)\times SO(3)\to {\mathbb R}$
on $SO(3)$
is said to be {\it bi-invariant} provided
$d(CAD,CBD)=d(A,B)$ for any $A,B,C,D\in SO(3)$.
The Lie group $SO(3)$ supports the unique (up to scale) bi-invariant metric
$$d(A,B)=-{1\over 2}~{\rm trace}\bigl ( {\rm log}(AB^t)\bigr )^2,$$
where the trace of a matrix is the sum of its diagonal entries and the logarithm
is the matrix logarithm
\cite{Liegroups}.

\begin{proposition} \label{traces}
For any $A_1,A_2\in SO(3)$, we have $d(A_1,I)<d(A_2,I)$ if and only if
${\rm trace} (A_2)<~{\rm trace}(A_1)$, where $d$ is the unique bi-invariant metric on $SO(3)$.
\end{proposition}

\begin{proof}
For any $A\in SO(3)$, there is $B\in SO(3)$ so that
$$BAB^t=\left (
\begin{array}{ccc}
\hskip2ex{\rm cos}~\theta&{\rm sin}~\theta&0\\
-{\rm sin}~\theta&{\rm cos}~\theta&0\\
0&0&1\\
\end{array}
\right)$$ for some angle $0\leq\theta\leq\pi$, cf.\ \cite{Liegroups}.
It follows from bi-invariance that
$$d(A,I)=d(BAB^t,BIB^t)=d(BAB^t,I)=d(BAB^{-1},I),$$
 i.e., distance to $I$ is a conjugacy invariant, and from the
formula for $d$ that $d(A,I)$ is a monotone increasing function of $\theta$.
On the other hand, ${\rm trace}(A)={\rm trace}(BAB^{-1})={\rm trace}(BAB^{t})=1+2~{\rm cos}~\theta$ is a monotone
decreasing function of $\theta$ which is also a conjugacy invariant, and the result follows.
\end{proof}

A (positively oriented) {\it 3-frame} is an ordered triple ${\mathcal F}=(\vec u_1,\vec u_2,\vec u_3)$ of three mutually perpendicular unit vectors in ${\mathbb 
R}^3$ so that $\vec u_3=\vec u_1\times\vec u_2$.  For example, the standard unit basis vectors $(\vec i,\vec j,\vec k)$ give a standard 3-frame.

\begin{proposition}\label{3-frames}
An ordered pair ${\mathcal F}=(\vec u_1,\vec u_2,\vec u_3)$ and ${\mathcal G}=(\vec v_1,\vec v_2,\vec v_3)$ of 3-frames uniquely determines an element $D\in 
SO(3)$, where $D\vec u_i=\vec v_i$, for $i=1,2,3$.  Furthermore, the trace of $D$ is given by
$\vec u_1\cdot \vec v_1+\vec u_2\cdot \vec v_2+\vec u_3\cdot \vec v_3$,
where $\cdot$ is the usual dot product of vectors in ${\mathbb R}^3$.
\end{proposition}

\begin{proof}
Express
$$\aligned
\vec u_i &= a_{1i}\vec i + a_{2i}\vec j + a_{3i}\vec k,\\
\vec v_i &= b_{1i}\vec i + b_{2i}\vec j + b_{3i}\vec k,\\
\endaligned$$
for $i=1,2,3$, as linear combinations of $\vec i,\vec j,\vec k$.  The matrices $A=(a_{ij})$ and $B=(b_{ij})$ thus map $\vec i,\vec j,\vec k$ to
$\vec u_1,\vec u_2,\vec u_3$ and
$\vec v_1,\vec v_2,\vec v_3$, respectively.
It follows that the matrix $D=BA^{-1}$ indeed has the desired properties.  If $D'$ is another such matrix, then $D^{-1}D'$ must fix each vector $\vec u_1,\vec 
u_2,\vec u_3$, and hence must agree with the identity proving the first part.  For the second part since trace is a conjugacy invariant,
we have
$${\rm trace}(BA^{-1}) = ~{\rm trace} (A^{-1}B) = ~{\rm trace} (A^t B)=\sum_{i=1}^3 \vec u_i\cdot \vec v_i$$ as was claimed.
\end{proof}

Suppose that $\Gamma$ is a graph.  An {\it $SO(3)$ graph connection} on $\Gamma$ is the assignment of an element
$A_f\in SO(3)$ to each oriented edge $f$ of $\Gamma$ so that
the matrix associated to the reverse of $f$ is the transpose of $A_f$.
Two such assignments $A_f$ and $B_f$ are regarded as equivalent if there is an assignment
$C_u\in SO(3)$ to each vertex $u$ of $\Gamma$ so that
$A_f= C_u B_f C_w^{-1}$, for each oriented edge $f$ of $\Gamma$ with initial point $u$ and terminal point $w$.
An $SO(3)$ graph connection on $\Gamma$ determines an isomorphism class of  flat
principal $SO(3)$ bundles over $\Gamma$, cf.\ \cite{Connections}.  Given an oriented edge-path $\gamma$ in $\Gamma$ described by consecutive oriented edges 
$f_0-f_1-\cdots -f_{k+1}$, where
the terminal point of $f_i$ is the initial point of $f_{i+1}$, for $i=0,\ldots ,k$, the {\it parallel transport operator} of the $SO(3)$ graph connection along 
$\gamma$ is given by the matrix product
$\rho (\gamma )=A_{f_0} A_{f_1}\cdots A_{f_{k}}\in SO(3)$.  In particular, if the terminal point of $f_{k}$ agrees with the initial point of $f_0$ so that 
$\gamma$ is a closed oriented edge-path, then ${\rm trace}(\rho(\gamma))$ is the {\it holonomy} of the graph connection along $\gamma$ and is well-defined on 
the equivalence class of graph connections.

\subsection{Modeling the backbone}\label{backbonesec}

In this section, we shall define our model $T(P)$ for the backbone of a polypeptide structure $P$.
To this end, consider the fatgraph building block depicted in Figure~\ref{fig:buildingblock}, which consists of a
horizontal segment and two vertical segments joined to distinct interior points of the horizontal segment, the vertical segment  on the left lying above and
on the right below the horizontal segment.  Each such building block represents a peptide unit, as is also indicated in the figure, where the left and right 
endpoints of the horizontal segment represent $C_i^\alpha$ and $C_{i+1}^\alpha$ and are labeled by the corresponding residue $R_i$ and $R_{i+1}$, respectively, 
the left and right trivalent vertices represent $C_i$ and $N_{i+1}$, respectively, and the endpoints of the vertical segments above and below the horizontal 
segment represent  $O_i$ and $H_{i+1}$, respectively, or in the case that $R_{i+1}$ is Proline, the endpoint of the vertical segment below the horizontal 
segment
instead represents the non-alpha carbon atom bonded to $N_{i+1}$ in the Proline ring.  In the case of cis-Proline, a more geometrically accurate building block 
would
have the vertical segment on the right also lying above the horizontal segment as indicated by the skinny line in the figure, but we nevertheless use a single 
building block in all cases for convenience.

\begin{figure}[!h]
\begin{center}
\epsffile{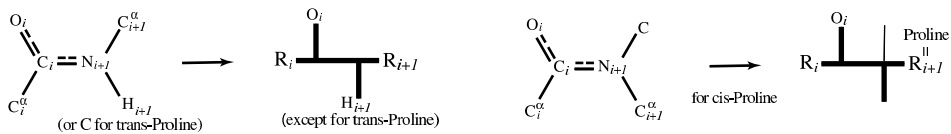}
\caption{Fatgraph building block}
\label{fig:buildingblock}
\end{center}
\end{figure}

Fix a polypeptide structure $P$ and start by defining a fatgraph $T'(P)$ as the concatenation of $L-1$ copies of this fatgraph building block, where the two 
univalent vertices representing
$C_{i+1}^\alpha$ are identified so that the two incident edges are combined to form a single horizontal edge of $T'$ called the $(i+1)$st {\it alpha carbon 
linkage}, for $i=1,\ldots ,L-2$ as illustrated in Figure~\ref{fig:concatenate}; let us also refer to the horizontal edges incident on the vertex corresponding
to $C^\alpha_1$ and $C_L^\alpha$ as the first and $L$th {\it alpha carbon linkages}, respectively, so the $i$th alpha carbon linkage is naturally
labeled by the amino acid $R_i$, for $i=1,\ldots ,L$.
Thus, $T'(P)$ consists of a long horizontal segment composed of $2L-1$ horizontal edges, $L$ of which are alpha carbon linkages and $L-1$ of which correspond to 
peptide bonds, with $2L-2$ short vertical edges attached to it alternately lying above and below the long horizontal segment.  We shall define the fatgraph 
$T(P)$ by specifying twisting on the alpha carbon linkages of $T'(P)$.

\begin{figure}[!h]
\begin{center}
\epsffile{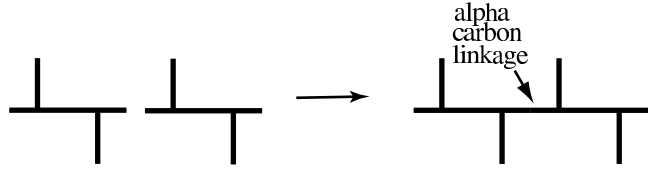}
\caption{Concatenating fatgraph building blocks}
\label{fig:concatenate}
\end{center}
\end{figure}

\begin{construction}\label{frames}
Associate a 3-frame ${\mathcal F}_i=(\vec u_i,\vec v_i,\vec w_i)$ to each peptide unit using input iii) by setting
$$\aligned
\vec u_i&={1\over {|\vec x_i|}}~~\vec x_i,\\\\
\vec v_i&={1\over{|\vec y_i -{{(\vec u_i\cdot \vec y_i}})~\vec u_i|}}~~\biggl (\vec y_i -{{(\vec u_i\cdot \vec y_i}})~\vec u_i\biggr ),\\\\
\vec w_i&=\vec u_i\times \vec v_i,\\
\endaligned$$
for $i=1,\ldots, L-1$, where $|\vec t|$ denotes the norm of the vector $\vec t$.
\end{construction}

Thus, $\vec u_i$ is the unit displacement vector from $C_i$ to $N_{i+1}$,
$\vec v_i$ is the projection of $\vec y_i$ onto the specified perpendicular of $\vec u_i$ in the plane of the peptide unit,
and $\vec w_i$ is the specified normal vector to this plane.

According to Proposition~\ref{3-frames}, there is a unique element $A_i\in SO(3)$ mapping ${\mathcal F}_i$ to ${\mathcal F}_{i+1}$, for $i=1,\ldots ,L-2$.  
Define the {\it backbone graph connection}
on the graph underlying $T'(P)$ to take value $I$ on all oriented edges except on the $i$th alpha carbon linkage oriented from its endpoint representing $N_{i}$ 
to its endpoint representing $C_{i}$, it takes value $A_{i-1}$, for $i=2,\ldots , L-1$.

We shall discretize the backbone graph connection to finally define the backbone fatgraph model $T(P)$.  To this end, in addition to the 3-frames
${\mathcal F}_i=(\vec u_i,\vec v_i,\vec w_i)$ of Construction~\ref{frames}, we consider also the 3-frames
${\mathcal G}_i=(\vec u_i,-\vec v_i,-\vec w_i)$, which correspond to simply turning ${\mathcal F}_i$ upside down by rotating through 180 degrees in three-space 
about the line containing $C_i$ and $N_{i+1}$, for $i=1,\ldots ,L-1$.  Again, by the first part of Proposition \ref{3-frames}, there is a unique element  
$B_i\in SO(3)$ taking ${\mathcal F}_i$
to ${\mathcal G}_{i+1}$.  By construction, $A_{i}$ also takes ${\mathcal G}_{i}$ to ${\mathcal G}_{i+1}$, and $B_{i}$ takes ${\mathcal G_{i}}$ to ${\mathcal 
F}_{i+1}$.

\begin{construction}\label{backbone}
For any polypeptide structure $P$, define the fatgraph $T(P)$ derived from
$T'(P)$ by taking twisting only on certain of the alpha carbon linkages, where the $(i+1)$st alpha carbon   linkage is twisted if and only if
$$\begin{cases}
d(I,B_i)\leq d(I,A_i),& {\rm if}~R_{i+1}~{\rm is~not~cis}\hskip -.5ex-\hskip -.6ex{\rm Proline};\cr
d(I,B_i)\geq d(I,A_i),& {\rm if}~R_{i+1}~{\rm is~cis}\hskip -.5ex-\hskip -.6ex{\rm Proline},\cr
\end{cases}$$
for $i=1,\ldots ,L-2$, where $d$ is the unique bi-invariant metric on $SO(3)$.
\end{construction}

\begin{corollary}The $(i+1)$st alpha carbon linkage of the backbone model $T(P)$ is twisted if and only if
$$\begin{cases}
\vec v_{i}\cdot\vec v_{i+1}~+~\vec w_{i}\cdot \vec w_{i+1}\leq0,& {\rm if}~R_{i+1}~{\rm is~not~Proline~or}~\vec y_i\cdot \vec z_i\geq 0;\cr
\vec v_{i}\cdot\vec v_{i+1}~+~\vec w_{i}\cdot \vec w_{i+1}\geq0,& {\rm if}~R_{i+1}~{\rm is~Proline~and}~\vec y_{i}\cdot \vec z_{i}<0,\cr
\end{cases}$$
for $i=1,\ldots ,L-2$.
\end{corollary}

\begin{proof}
According to Proposition \ref{traces}, $d(A_i,I)\leq d(B_i,I)$ if and only if
${\rm trace}(B_i) \leq ~{\rm trace}(A_i)$.   According to the second part of Proposition \ref{3-frames}, we have
$$\aligned
{\rm trace}(A_i)&=\vec u_{i}\cdot \vec u_{i+1}~+~\vec v_{i}\cdot \vec v_{i+1} ~+~\vec w_{i}\cdot \vec w_{i+1},\\
{\rm trace}(B_i)&=\vec u_{i}\cdot \vec u_{i+1}~-~\vec v_{i}\cdot \vec v_{i+1}~-~\vec w_{i}\cdot \vec w_{i+1},\\
\endaligned$$
so that ${\rm trace}(A_i)-{\rm trace}(B_i)=2(\vec v_{i}\cdot \vec v_{i+1} +\vec w_{i}\cdot \vec w_{i+1})$.

Thus, if $R_{i+1}$ is not Proline, then we twist the $(i+1)$st alpha carbon linkage if and only if ${\mathcal F}_i$ is closer to ${\mathcal G}_{i+1}$ than it is 
to
${\mathcal F}_{i+1}$ in the sense that $d(I,A_i)\geq d(I,B_i)$, and this is our natural discretization of the backbone graph connection in Construction 
\ref{backbone} in this case.
If $R_{i+1}$ is Proline, then it is in the cis-conformation if and only if $\vec y_i\cdot \vec z_i<0$,
so we twist the $(i+1)$st alpha carbon linkage for cis-Proline only if $d(I,A_i)\leq d(I,B_i)$.  To see that this is the natural discretization of the backbone 
graph connection in this case,
notice that the 3-frame ${\mathcal F}_i$ in Construction~\ref{frames} is determined using the displacement vectors $\vec x_i$ from $C_i$ to $N_{i+1}$ and
$\vec y_i$ from $C_{i}^\alpha$ to $C_i$, which are insensitive to whether $R_{i+1}$ is in the cis-conformation.  It is therefore only upon exiting a cis-Proline 
along the backbone
that the earlier determination should be modified since the latter displacement vector should be replaced by its antipode.
\end{proof}

\begin{figure}[!h]
\begin{center}
\epsffile{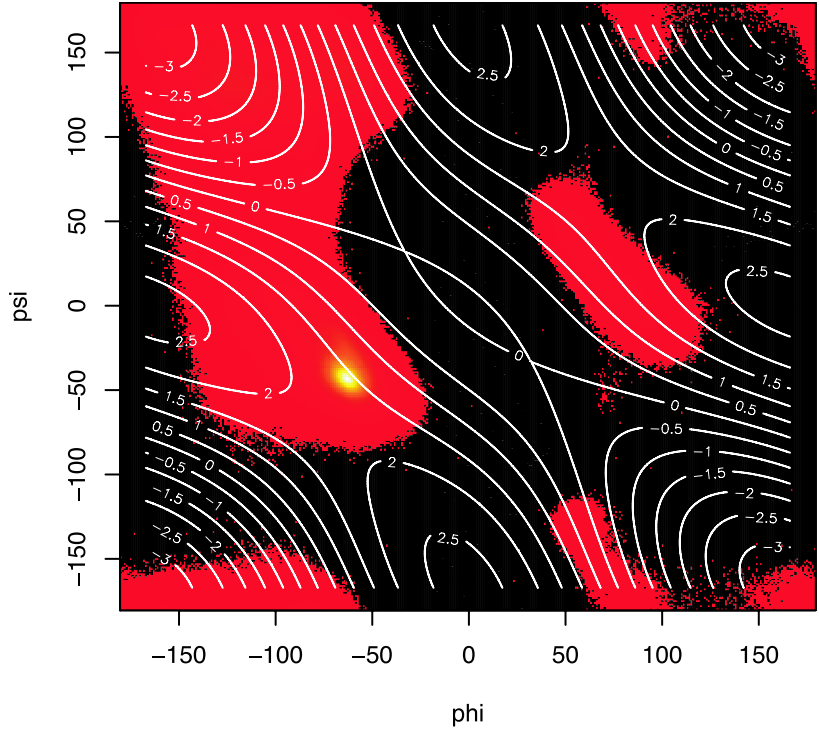}
\caption{Level sets of ${\rm trace}(A)-{\rm trace}(B)$ on a Ramachandran plot}
\label{fig:rama}
\end{center}
\end{figure}

Define the {\it flip sequence} of $G(P)$ to be the
word in the alphabet $\{ {\rm F,N}\}$ whose $i$th letter is N if and only if the $(i+1)$st alpha carbon linkage is untwisted, for $i=1,\ldots ,L(G)-2$.  The 
flip sequence thus gives a discrete invariant  assigned to each alpha carbon linkage derived from the conformational geometry along the backbone.  The flip 
sequence can be determined directly from the conformational angles along the backbone using the following result.

\begin{proposition}\label{AandB}
Under the idealized geometric assumptions of tetrahedral angles among bonds at each alpha carbon atom and $120$-degree angles between bonds within a peptide unit, the matrix $A=A_i$ in Construction \ref{backbone} can be calculated in terms of the conformational angles $\varphi=\varphi_i$ and $\psi=\psi_i$ as follows:

$$A=B_3(\varphi) B_2(\varphi+\psi)
 \begin{pmatrix}
 -{1\over 2}&{\sqrt{3}\over 2}&~~0\\
 {\sqrt{3}\over 2}&{1\over 2}&~~~0\\
 0&0&-1\\
\end{pmatrix},
$$
where
\begin{align*}
B_3(\varphi)=
\begin{pmatrix}
{2\over 3}-{{C^2}\over 3}+{{S^2}\over 6}
&-2\bigl [{{\sqrt 2 C}\over 3}+{{S^2}\over{4\sqrt{3}}}\bigl ]
&2\bigl [{{CS}\over{2\sqrt{3}}}-{S\over{3\sqrt{2}}}\bigl ]\\
2\bigl [{{\sqrt{2}C}\over 3}-{{S^2}\over{4\sqrt{3}}}\bigl ]
&{2\over 3}-{{C^2}\over 3}-{{S^2}\over 6}
&-2\bigl [{{CS}\over 6}+{S\over\sqrt{6}}\bigl ]\\
2\bigl [{{CS}\over{2\sqrt{3}}}+{S\over{3\sqrt{2}}}\bigl ]
&2\bigl [{S\over\sqrt{6}}-{{CS}\over 6}\bigl ]
&{2\over 3}+{{C^2}\over 3}-{{S^2}\over 3}
\end{pmatrix},~{\rm for}~~
\aligned
C&={\rm cos}~\varphi,\\
S&={\rm sin}~\varphi,\\
\endaligned
\end{align*}

\begin{align*}
B_2(\varphi+\psi)=
\begin{pmatrix}
1-\frac32 S^2
&\frac{\sqrt{3}}{2} S^2
&\sqrt{3}CS\\
\frac{\sqrt{3}}{2}S^2
&1-\frac12S^2
&-CS\\
-\sqrt{3}CS
&CS
&1-2S^2\\
\end{pmatrix},~{\rm for}~~
\aligned
C&={\rm cos}~{{\varphi+\psi}\over 2},\\
S&={\rm sin}~{{\varphi+\psi}\over 2}.\\
\endaligned
\end{align*}
Explicitly, this is the representative $A=A_i$ in its conjugacy class
for which the 3-frame vectors $\vec u_i=\vec i, \vec v_i= \vec j, \vec w_i=\vec k$
in Construction \ref{frames}
are given by the standard unit basis vectors.
\end{proposition}

\begin{proof}
Let $\xi$ be an angle and $\vec v$ be a non-zero vector in $\mathbb{R}^3$. We denote by $(\xi,\vec v)$ the linear transformation $\mathbb{R}^3\rightarrow\mathbb{R}^3$ which rotates $\mathbb{R}^3$ through the angle $\xi$ around the line spanned by $\vec v$ in the right-handed sense in the direction of $\vec v$.  By following the standard 3-frame along the backbone in the natural way one bond at a time, we find
\begin{align*}
A=B_6(\varphi,\psi)B_5(\varphi,\psi)B_4(\varphi,\psi)B_3(\varphi)B_2(\varphi)B_1(\pi/3)
\end{align*}
where
\begin{align*}
B_1(\xi)=(\xi,\vec k),\quad B_2(\varphi)=(\varphi,B_1(\pi/3)\vec i ),\quad B_3(\phi)=(\pi-\theta,B_2(\phi)\vec k),
\end{align*}
\begin{align*}
B_4(\varphi,\psi)=(\psi,B_3(\varphi)B_1(\pi/3)\vec i),\quad B_5(\varphi,\psi)=(2\pi/3,-B_4(\varphi,\psi)B_3(\varphi)B_2(\varphi)\vec k),
\end{align*}
\begin{align*}
B_6(\varphi,\psi)=(\pi,B_5(\varphi,\psi)B_4(\varphi,\psi)B_3(\varphi)B_2(\varphi)B_1(\pi/3)\vec j),
\end{align*}
and where $\theta=2\arctan(\sqrt{2})$ is the tetrahedral angle $\approx$ 109.5 degrees, for which ${\rm cos}~\theta =-{1\over 3}$.

We observe that
$$B_4(\varphi,\psi)B_3(\varphi)=B_3(\varphi)B_2(\psi)$$ whence
$$B_4(\varphi,\psi)B_3(\varphi)B_2(\varphi)=B_3(\varphi)B_2(\varphi+\psi),$$
and therefore
$$A=B_6(\varphi,\psi)~B_3(\phi)~B_2(\varphi+\psi)~B_1(-\pi/3).$$
Setting $B_0=(\pi  ,\vec j)$, we conclude
\begin{align*}
A=B_3(\varphi)B_2(\varphi+\psi)B_1(-\pi/3)B_0,.
\end{align*}
which devolves after some computation to the given expression.
\end{proof}

\begin{remark}\label{ramaplot}
It is interesting to graph the level sets of ${\rm trace}(A)-{\rm trace}(B)$ on the Ramachandran
plot, i.e., the plot of pairs of conformational angles $(\varphi_i,\psi_i)$,  for the entire CATH database
\cite{CATH}
using Proposition \ref{AandB} as depicted in Figure~\ref{fig:rama},
where the matrix $B=B_i$ of Construction~\ref{frames} is obtained from $A=A_i$ in Proposition \ref{AandB} 
by pre-composing it with rotation by $\pi$ about $\vec i$.
In particular, the zero level set fairly well avoids highly populated regions, so the case 
of near equality
in Construction \ref{backbone} is a relatively rare phenomenon for proteins.\footnote{Indeed, further scrutiny of detail in Figure \ref{fig:rama}, which is not depicted, shows that the zero level set does penetrate into conformations of ``beta turns of types II and VI'', cf. the discussion of Figure \ref{fig:alphanbeta}.} 
\end{remark}

\subsection{Modeling hydrogen bonds}\label{hydrogen}

The fatgraph model $T(P)$ of the backbone of a polypeptide structure $P$ defined in the previous section is here completed to our fatgraph model $G(P)$.  Just 
as in the previous section,
we shall first define another fatgraph $G'(P)$ from which $G(P)$ is derived by further twisting certain of its edges.  As described in the previous section, 
$T(P)$
consists of a long horizontal segment, certain of whose alpha carbon linkages are twisted, together with small vertical segments alternately lying above and 
below the long
horizontal segment, where the $(i+1)$st alpha carbon linkage is labeled by its corresponding amino acid $R_{i+1}$, for $i=1,\ldots, L$.  The endpoints of the 
vertical segments above and below the horizontal segment respectively represent the atoms $O_i$ and $H_{i+1}$ except for the vertical segments below the 
horizontal segment preceding an alpha carbon linkage labeled by Proline, whose endpoint represents the non-alpha carbon atom bonded to $N_{i+1}$ in the 
corresponding Proline ring, for $i=1,\ldots ,L-1$.

\begin{figure}[!h]
\begin{center}
\epsffile{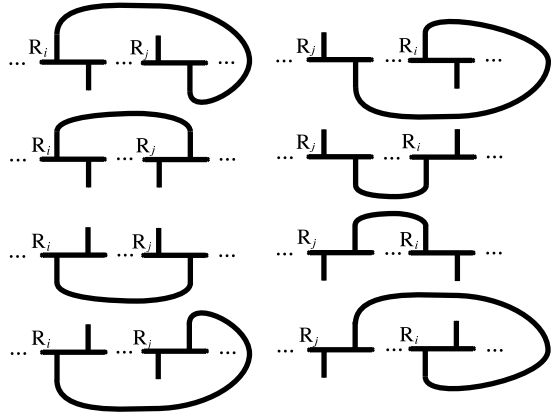}
\caption{Adding edges to $T(P)$ for hydrogen bonds}
\label{fig:addhydrogen}
\end{center}
\end{figure}

\begin{construction}\label{constructfatgraph'}
For each $(i,j)\in{\mathcal B}$ in input ii), adjoin an edge to $T(P)$ without introducing new vertices connecting the endpoints of short vertical segments 
corresponding to $H_i$ and $O_j$
to produce a fatgraph denoted $G'(P)$.
\end{construction}

\noindent See Figure~\ref{fig:addhydrogen}.  It is important to emphasize that the relative positions of these added edges corresponding to hydrogen bonds other 
than their endpoints are completely immaterial to the strong equivalence class of $G'(P)$.  The edges of $T(P)$ corresponding to the non alpha carbon atoms in a 
Proline rings are never hydrogen bonded in our model.

To complete the construction of $G(P)$, it remains only to determine which edges of the fatgraph $G'(P)$ are twisted.
To this end, suppose that  $(i,j)\in{\mathcal B}$ in input ii).  According to our enumeration of peptide units, $H_i$ occurs in peptide unit $i-1$ and $O_j$ 
occurs in peptide unit $j$, and there are corresponding 3-frames
$$\aligned
{\mathcal F}_{i-1}&=(\vec u_{i-1},\vec v_{i-1},\vec w_{i-1}),\\
{\mathcal F}_{j}&=(\vec u_j,\vec v_j,\vec w_j),\\
{\mathcal G}_{j}&=(\vec u_j,-\vec v_j,-\vec w_j),\\
\endaligned$$
from Construction \ref{frames}.

\begin{construction}\label{constructfatgraph}
As before by the first part of Proposition \ref{3-frames}, there are unique $D_{i,j},E_{i,j}\in SO(3)$ taking ${\mathcal F}_{i-1}$ to ${\mathcal F}_j,{\mathcal 
G}_j$ respectively.
An edge of $G'(P)$ corresponding to the hydrogen bond $(i,j)\in{\mathcal B}$ is
twisted in $G(P)$ if and only if
$$d(I,E_{i,j})\leq d(I,D_{i,j}),$$
where $d$ is the unique bi-invariant metric on $SO(3)$.
\end{construction}

As before, a short computation gives:

\begin{corollary} The edge of $G(P)$ corresponding to the hydrogen bond $(i,j)\in{\mathcal B}$ is
twisted if and only if $\vec v_{i-1}\cdot \vec v_i +\vec w_{i-1}\cdot \vec w_j\leq 0$.
\end{corollary}

\begin{remark}\label{holonomy}
The backbone graph connection on the graph underlying $T(P)$ clearly has trivial holonomy since $T(P)$ is contractible.  It extends naturally to an $SO(3)$ 
graph connection on the graph
underlying $G(P)$, where to the oriented edge corresponding to the hydrogen bond connecting $N_i-H_i$ and $O_j$, we assign the unique element of $SO(3)$, whose 
existence is guaranteed by
Proposition~\ref{3-frames}, which maps ${\mathcal F}_{i-1}$ to ${\mathcal F}_j$, for $i=2,\ldots ,L-2$.  This graph connection on $G(P)$ also has trivial 
holonomy by construction.  Our fatgraph model $G(P)$ arises from a discretization of this $SO(3)$ graph connection giving a ${\mathbb Z}/2$ graph connection, 
where the oriented edges with non-trivial holonomy are the twisted ones, and this ${\mathbb Z}/2$ graph connection on the graph underlying $G(P)$ typically does 
not have trivial holonomy.
\end{remark}

\subsection{The basic model and its extensions}\label{generalmodel}

The previous section completed the
definition of our basic fatgraph model $G(P)$ of a polypeptide structure $P$.
Notice that hydrogen bonds and alpha carbon linkages are treated in precisely the same manner in this construction.

A crucial point in practice is that the polypeptide structure itself depends upon data
which must be considered as idealized for various reasons: proteins actually occur in several
closely related conformations, varying under thermal fluctuations for example, whose sampling is corrupted by experimental uncertainties as well as errors.
The fatgraph $G(P)$
must therefore not be taken as defined absolutely, but rather as defined only in some statistical sense as a family of fatgraphs $\{ G(P):P\in{\mathcal P}\}$ 
based on
a collection ${\mathcal P}$ of polypeptide structures which differ from one another by a small number of such idealizations, uncertainties, or errors.  
Properties of the fatgraph $G(P)$ that we can meaningfully assign to the polypeptide structure $P$ must be nearly constant on ${\mathcal P}$ leading to the 
notion of ``robustness'' of invariants of $G(P)$ as descriptors of $P$, which is
discussed  in Section~\ref{protinvariants}.  Nevertheless, the construction of our model has been given based on the inputs  above regarded as exact and 
error-free.

In particular, there is the tacit assumption that there is never equality in the determination of whether to twist in Constructions~\ref{backbone}.
In practice,
$\vec v_{i}\cdot\vec v_{i+1}+\vec w_{i}\cdot \vec w_{i+1} =0$
never occurs exactly, but there is the real possibility that this condition {\sl nearly holds}, that is, we cannot
reliably determine whether to twist if $|\vec v_{i}\cdot\vec v_{i+1}+\vec w_{i}\cdot \vec w_{i+1}|$ is below some small threshold because of experimental 
uncertainty, cf.
Remark \ref{ramaplot}.
There are similar issues in the specification of which hydrogen bonds exist in input ii)
based upon the possibly problematic exact atomic locations from which the electrostatic potentials are inferred as well as whether to twist in 
Construction~\ref{constructfatgraph}.

However, there is the following control over the topological type of $F(G(P))$, which will be the basis for several  of the robust invariants of
fatgraphs and resulting meaningful descriptors of polypeptides studied in Section~\ref{protinvariants}.

\begin{corollary}\label{risrobust}
Let $P,P'$ be polypeptide structures with the same inputs i) but differing in inputs ii-iii) in the determinations of the existence
of $m$ hydrogen bonds and of the twisting of $n$ alpha carbon linkages or hydrogen bonds.  Then
$|r(G(P))-r(G(P'))|\leq m+n$.
\end{corollary}

\begin{proof}
This is an immediate consequence of Proposition~\ref{changetwists}.
\end{proof}

There are several generalizations of the basic fatgraph model $G(P)$ of a polypeptide structure.  As already mentioned, we might specify energy thresholds 
$E_{-}<E_{+}<0$ and
demand that the potential energy of a hydrogen bond lie in the range between $E_{-}$ and $E_{+}$ in order that it be regarded as a hydrogen bond to include in 
input ii) so as to produce a
fatgraph denoted $G_{E_-,E_+}(P)$.  We shall
describe in Section~\ref{globules} certain experiments with proteins using various such energy thresholds.

One may also model bifurcated hydrogen bonds and allow
hydrogen or oxygen atoms in the peptide units to participate in at most $\beta\geq 1$ hydrogen bonds by altering
the fagraph building block in Figure~\ref{fig:buildingblock} by replacing the univalent vertices
representing hydrogen and oxygen atoms by vertices of valence $\beta +1$.  Different  valencies less than $\beta +1$ for oxygen and hydrogen can be implemented
with this single building block by appropriately imposing different constraints in input ii).  Natural fattenings on these new vertices representing hydrogen or 
oxygen
atoms are determined as follows: project centers of partners in bonding into the plane of the peptide unit with origin at the center of the corresponding 
nitrogen or carbon atom, respectively, where the positive $x$-axis contains the bond axis of the incident peptide bond, and take these projections in the 
ordering of increasing argument.

Our definition of polypeptide structure assumes that there are no atoms missing along the backbone, and this is actually somewhat problematic in practice.  A 
useful aspect of the methods in Section~\ref{backbonesec} is that such gaps present no essential difficulty since an edge connecting fatgraph building blocks 
can just as well be taken to represent a gap between peptide units as to represent an alpha carbon linkage as in our
model articulated before.  The determination of twisting on these new gap edges is just as in Construction~\ref{backbone}, but now the 3-frames in this 
construction do not correspond to consecutive peptide units.

A more profound extension of the method is to use the bi-invariant metric on $SO(3)$ to give finer discretizations of the $SO(3)$ graph connection on $G(P)$ 
discussed in
Remark~\ref{holonomy}.  For example, rather than our ${\mathbb Z}/2$ graph connection modeled by fatgraphs, one can easily implement the analogous construction 
of an ${\mathbb Z}/n$ graph connection based on the natural extensions of Constructions~\ref{backbone} and \ref{constructfatgraph} modeled by graphs with 
fattenings and ${\mathbb Z}/n$-colorings.  These ``rotamer fatgraphs'' capture the ``protein rotamers'' which are highly studied in the biophysics literature.

A still more profound innovation rests on the observation that
our techniques are of greater utility and can be adapted to model essentially any molecule since 3-frames can analogously be associated to any bond axis.  One 
might thus
model entire amino acids themselves as rotamer fatgraphs to give a truly realistic model of a polypeptide.

Furthermore, the discussion thus far has concentrated on molecules at equilibrium, and one might instead regard the fatgraph or rotamer fatgraph as a dynamic 
model by taking time or temperature dependent inputs i-iii).

\section{Robust polypeptide descriptors}\label{protinvariants}

We have described in the previous sections the fatgraph $G(P)$ of a polypeptide structure $P$ with simple hydrogen bonding determined by inputs i-iii) based 
upon
specified energy thresholds.  With the understanding that the input data can be problematic due to errors and experimental indeterminacies, we must consider the 
fatgraph as defined only in a statistical sense, where a family of fatgraphs arises from a collection ${\mathcal P}\ni P$ of polypeptide structures which differ 
from $P$ by a small number of such errors or indeterminacies.  As such, only certain properties of the fatgraph  $G(P)$ can meaningfully be assigned as 
descriptors of $P$, namely, those properties which do not vary significantly over the various polypeptide structures in ${\mathcal P}$.
In this section, we shall first formalize this notion of meaningful properties of fatgraphs, and then describe and discuss a myriad of such polypeptide 
descriptors.

Let ${\mathcal G}$ denote the collection of all strong equivalence classes of fatgraphs $G(P)$ arising from non-empty polypeptide structures $P$.  We may 
perform the following
modifications to any $G\in{\mathcal G}$ leaving all other data unchanged:

\smallskip

\noindent {\bf Mutation i)}~change the color of one alpha carbon linkage of $G$;

\smallskip

\noindent {\bf Mutation ii)}~change the color of
one edge of $G$ corresponding to a hydrogen bond;

\smallskip

\noindent {\bf Mutation iii)}~add or delete an untwisted edge of $G$ corresponding to a hydrogen bond;

\smallskip

\noindent{\bf Mutation  iv)}~replace a fatgraph building block of $G$ by two building blocks
connected by an untwisted alpha carbon linkage, where any edges corresponding to hydrogen bonds incident on the
original building block are connected to the replacement building block that occurs first along the
backbone from $N$ to $C$ termini, and the reverse of this operation.

\smallskip

Suppose that $X$ is some set with metric
$\rho$.  We say that a function $\nu:{\mathcal G}\to X$ is  {\it $\kappa$-robust of radius $Q$ on ${\mathcal H}\subseteq{\mathcal G}$}, where $\kappa\geq 0$ is 
real and
$Q\geq 0$ is an integer,
if $\rho(\nu(G),\nu(G'))\leq q\kappa$ whenever $G'$ arises from $G\in{\mathcal H}$ by
a sequence
$$G=G_0-G_1-\cdots -G_q=G',~{\rm with}~q\leq Q,$$
where $G_{j+1}$ arises from $G_j$ by a single mutation of type i-iv), for $j=0,\ldots ,q-1$.
If $\nu$ is $\kappa$-robust of infinite radius on all of ${\mathcal G}$, then we say simply that $\nu$ is $\kappa$-robust.

By definition if $X$ supports operations of addition and scalar multiplication and if $\nu$ is $\kappa$-robust of radius $Q$ on ${\mathcal H}$, then for any $\alpha\in{\mathbb R}$, $\alpha\nu$ is $\alpha \kappa$-robust of radius 
$Q$ on ${\mathcal H}$, and furthermore, if
$\nu'$ is $\kappa'$-robust of radius $Q'$ on ${\mathcal H}'$, then $\nu\pm\nu'$ is $(\kappa+\kappa')$-robust of radius ${\rm min}(Q,Q')$ on ${\mathcal 
H}\cap{\mathcal H}'$.

It is only the $\kappa$-robust functions $\nu$ of reasonably large radius $Q$ and sufficiently small
value of $\kappa$ on ${\mathcal H}\subseteq{\mathcal G}$ which are significant characteristics of polypeptide structures whose fatgraphs $G$ lie in ${\mathcal 
H}$.  This is because
a combination of mutations arising from $q\leq Q$ errors or indeterminacies of the input data
then affects the value of $\nu(G)$ by an amount bounded by $q\kappa$, which must be small compared
to the value of $\nu(G)$.

It is clear that any two fatgraphs arising from a non-empty polypeptide structure are related by a finite sequence of mutations i-iv).  By assigning a penalty 
of some non-zero magnitude to each type of mutation, the {\it mutation distance} between two such fatgraphs can be defined as the minimum sum of penalties 
corresponding to sequences of mutations relating them.  This gives a metric, albeit seemingly difficult to compute, on ${\mathcal G}$ itself, and we may regard 
two polypeptide structures as being similar if the mutation distance between their corresponding fatgraphs is small.  The assignment of fatgraph $G(P)$ to 
polypeptide structure $P$ is $\kappa$-robust by definition with this metric, where the parameter $\kappa$ is the largest penalty.

For several obvious numerical examples, the numbers $L(G)$ of residues and $B(G)$ of
hydrogen bonds of $G$ are $1$-robust, and the Euler characteristic
$\chi(G)$ of $G$ or $F(G)$ is likewise $1$-robust since
$\chi (G)=1-B(G)$. The numbers $v(G)=2L(G)-2$ of vertices and $e(G)=B(G)+2L(G)-3$ of edges of $G$
are therefore $2$- and $3$-robust respectively.  The number of twisted edges
corresponding to hydrogen bonds and the number of twisted alpha carbon linkages of $G$ are each also clearly $1$-robust.

With $X$ the set of all words of finite length in the alphabet $\{ F,N\}$ given the edit distance with unit operation cost \cite{Gusfield},
the flip sequence of $G$ is $1$-robust by definition.
In contrast, the plus/minus sequence of  the alternative model $K(P)$ in Appendix \ref{alternative} as a word in the alphabet $\{ +,-\}$ with the same metric
is not $\kappa$-robust  of radius greater than zero on
${\mathcal G}$ for any $\kappa$
since a single modification of type i) to $G$ can change
all the entries of the plus/minus sequence.

For another negative example with $X={\mathbb R}$, the genus $g(G)$ of $F(G)$ is not $\kappa$-robust of any radius greater than zero for any $\kappa$ on 
${\mathcal G}$
since a single modification of type  ii) on an untwisted $G$ can produce a fatgraph $G'$ with $F(G')$ non-orientable, and $|g(G)-g(G')|=[1+B(G)-r(G)]/2$.
In contrast, the modified genus is robust of infinite radius according to the following result.

\begin{proposition}\label{countem}
The number $r(G)$ of boundary components and the modified genus $g^*(G)$ of $F(G)$ are $1$-robust.
Moreover, the number of appearances in the flip sequence of $G$ of any fixed word of length $k$ in the alphabet $\{ 0,1\}$ is $k$-robust.
\end{proposition}

\begin{proof}
The function $r$ satisfies the required properties by Corollary~\ref{risrobust}, hence so too does
$g^*=(1+B-r)/2$.  The remaining assertion follows essentially by definition.
\end{proof}

Given a closed edge-path $\gamma$ on $G\in{\mathcal G}$, define the {\it peptide-length} of $\gamma$ to be the number of pairs of distinct peptide units visited 
by $\gamma$ and
define the {\it edge-length} of $\gamma$ to be the number of edges of $G$ traversed by $\gamma$, each counted with multiplicity.  For example, the dotted 
boundary components in Figure~\ref{fig:alphanbeta} that are characteristic of alpha helices and beta strands all have peptide-length 4 and various edge-lengths 
4,6,8.
Define the {\it peptide-length spectrum} ${\mathbb P}(G)$
and the {\it edge-length spectrum} ${\mathbb E}(G)$ of $G\in{\mathcal G}$ , respectively,
to be the unordered set of peptide-lengths and edge-lengths of boundary components of $F(G)$.  Let  $\bar{\mathbb P}(G)$
and  $\bar{\mathbb E}(G)$ denote their respective means.
It is worth pointing out that the preponderance of alpha helices and beta strands in practice heavily biases $\bar{\mathbb P}(G)$ towards 4.

Let $X$ denote the collection of all finite unordered collections of natural numbers.  The elements of
a member of $X$ may be ordered by increasing magnitude.
The distance between two such ordered
finite collections of natural numbers may then be defined by standard methods
\cite{Gusfield}, and this induces a metric on $X$ itself.  We may thus regard ${\mathbb P}$ and ${\mathbb E}$ as functions on ${\mathcal G}$ with values in the 
metric space $X$.  As in the proof of Corollary~\ref{risrobust}, these functions are $\kappa$-robust where the parameter $\kappa$ depends on the choice of 
metric.

\begin{lemma}\label{2robust}
Suppose that $\mu:{\mathcal G}\to {\mathbb Z}$ is $k$-robust of radius at least $Q$ on ${\mathcal G}$ and that $\nu:{\mathcal G}\to{\mathbb R}$ is 
$\kappa$-robust of radius $Q$ on
$${\mathcal H}=\{ G\in{\mathcal G}: \mu(G)>kQ~{\rm and}~\nu(G)+Q\kappa \leq[\mu(G)-kQ]^2\}.$$
Then $\nu(G)/\mu(G):{\mathcal G}\to{\mathbb R}$ is $(\kappa+k)$-robust of radius $Q$ on ${\mathcal H}$.
\end{lemma}

\begin{proof}Suppose that $G\in{\mathcal H}$ and that $G=G_0-G_1-\cdots -G_q=G'$ is a sequence as before, with $q\leq Q$.  First note that
$$\aligned
\nu(G_{i+1})&\leq \nu(G_0)+i\kappa~~{\rm and}~~
\mu(G_{i+1})&\hskip -2ex \geq\mu(G_0)-ki,
\endaligned$$
by hypothesis, and so
$${{\nu(G_{i+1})}\over{[\mu(G_{i+1})]^2}}\leq
{{\nu(G_0)+i\kappa}\over{[\mu(G_0)-ki]^2}}\leq
{{\nu(G_0)+Q\kappa }\over{[\mu(G_0)-kQ]^2}}\leq1$$
since $G_0\in{\mathcal H}$, for $i=0,\ldots ,p$.
Furthermore, we have that $|\nu(G_i)-\nu(G_{i+1})|\leq \kappa$ and $|\mu(G_i)-\mu(G_{i+1})|\leq k$, for each
$i=0,\ldots ,q-1$, and hence
$$\aligned
\biggl | {{\nu(G_{i})}\over{\mu(G_{i})}}- {{\nu(G_{i+1})}\over{\mu(G_{i+1})}}\biggr |&=
\biggl |
{{\mu(G_{i+1})\nu(G_i)-\mu(G_i)\nu(G_{i+1})}\over{\mu(G_i)\mu(G_{i+1})}}
\biggr |\\\\
&\leq\begin{cases}
{\kappa\over{|\mu(G_i)|}},&~{\rm if}~\mu(G_{i+1})=\mu(G_{i})\\
{\kappa\over{|\mu(G_i)|}}+k{{|\nu(G_{i+1})|}\over{[\mu(G_{i+1})]^2}},&~{\rm if}~\mu(G_{i+1})<\mu(G_{i})\\
{\kappa\over{|\mu(G_i)|}}+k{{|\nu(G_{i})|}\over{[\mu(G_{i})]^2}},&~{\rm if}~\mu(G_{i+1})>\mu(G_{i})\\
\end{cases}\\\\
&\leq \kappa+k.\\
\endaligned$$
The triangle inequality then gives
$$\biggl |
{{\nu(G)}\over{\mu(G)}}-{{\nu(\tilde G)}\over{\mu(\tilde G)}}
\biggr |\leq q(\kappa+k)$$
as required.
\end{proof}

\begin{proposition}\label{means}
The mean $\bar {\mathbb P}(G)$ of the peptide-length spectrum is $3$-robust of radius $Q$ on
$$\{ G\in{\mathcal G}:r(G)>Q~{\rm and}~L(G)+Q-1\leq {1\over 2}[r(G)-Q]^2\},$$
and the mean $\bar{\mathbb E}(G)$ of the edge-length spectrum is $7$-robust of radius $Q$ on
$$\{ G\in{\mathcal G}:r(G)>Q~{\rm and}~B(G)+2L(G)-3+6Q\leq [r(G)-Q]^2\}.$$
\end{proposition}

\begin{proof} Since each peptide unit occurs exactly twice in the union of all the boundary components, the sum of all the elements in ${\mathbb P}(G)$ is 
constant equal to $2[L(G)-1]$, which is $2$-robust according to earlier comments.  Since $\bar{\mathbb P}(G)=2[L(G)-1]/r(G)$ and $r(G)$ is 1-robust by 
Lemma~\ref{countem}, the first assertion follows from
Lemma~\ref{2robust}. Similarly, each edge occurs exactly twice in the union of all boundary components, so the sum of all the elements in ${\mathbb E}(G)$ is 
equal to
$2e(G)=2[B(G)+2L(G)-3]$, which is $6$-robust according to earlier comments.  The second assertion therefore likewise follows from Lemma~\ref{2robust}.
\end{proof}

Other notions of lengths of closed edge-paths in $G$ may also be useful.
For example, for each amino acid type, each boundary component of $F(G)$ visits a certain number of alpha carbon linkages labeled by amino acids of this type, 
and alternative notions of length arise by assigning weights to the various amino acids and taking the weighted sum over amino acids visited.  The robustness of 
these sorts of invariants seems difficult to analyze.

It is also worth pointing out that the underlying graph of the fatgraph $G(P)$ has its own related characteristics for any polypeptide structure $P$.  For 
example, there is an associated notion of length spectrum, namely, one or another of the notions of generalized length discussed before of the closed edge-paths 
or simple closed edge-paths on the graph.  Invariants of this type, which can be derived from the graph underlying the fatgraph, may also be of importance  in 
practice, and their robustness is based on the invariance of the underlying graph under the modifications i-ii).

The fatgraph $G$ is of a special type in that it has a ``spine'' arising from the backbone, namely, the long horizontal segment arising from the concatenation 
of horizontal segments in the fatgraph building blocks which was discussed in Section~\ref{backbonesec}.
This ``spined fatgraph'' admits a canonical ``reduction'' by serially removing each edge incident on a univalent vertex and
amalgamating the pair of edges incident on the resulting bivalent vertex into a single edge.  The graph underlying this reduced spined fatgraph is a ``chord 
diagram'', and there are countless  ``finite-type invariants associated with weight systems'' \cite{QuantumInvariants}, which could provide useful protein 
invariants whose robustness depends upon the choice of weight system.  See Section~\ref{closing} for a further discussion of related quantum invariants.

\section{First results}\label{globules}

\subsection{Aspects of implementation}\label{implementation}

In this section, we shall first make several practical remarks about the implementation in this paper of our methods for a protein from its PDB and DSSP files, cf.\ 
Section~\ref{polypeptides},
where we shall consider here only the model with simple hydrogen bonds, i.e., $\beta =1$, which depends upon energy thresholds
$E_{-}<E_{+}<0$ as follows.  In effect, we employ the standard methods of DSSP described in Section \ref{polypeptides}
to estimate electrostatic potentials of possible hydrogen bonds,
and we tabulate to hundredths of kcal/mole the two strongest such potentials in which each hydrogen or oxygen atom in a polypeptide unit participates.   Any 
such energies beyond our energy thresholds are then discarded.  Displacements of corresponding backbone atoms are used to discriminate between equal tabulated 
electrostatic potentials in order to derive a strict linear ordering on them: a hydrogen bond with energy $E$ between atoms at distance $\delta$ precedes a 
hydrogen bond with energy $E'$ between atoms at distance $\delta'$
if $E< E'$ or if $E=E'$ and $\delta\leq \delta'$, where $E=E'$ to hundredths of kcal/mole and $\delta=\delta'$ to thousandths of Angstroms never occurs in 
practice.  We finally greedily add to ${\mathcal B}$ in input ii) the hydrogen bonds in this linear ordering provided they do not violate the {\sl a priori} 
simple hydrogen bond assumption $\beta=1$.

Minor technical comments are that unspecified or missing residue types
are assumed not to be Proline for input i), atomic locations in the PDB with highest occupancy numbers are those used for determining input iii), and we take 
only
the first model in case there are several models in a PDB file.

Whenever there is a missing datum, for example the atomic location of a backbone atom in a PDB file, that is required for the algorithmic construction of the 
3-frame corresponding to its peptide unit,
we concatenate an associated fatgraph building block without twisting the alpha carbon linkage, and we prohibit any hydrogen bonding to its constituent edges.  
Such ``gap frames'' are included
for each problematic peptide unit.  A number of such gap frames may occur between two fatgraph building blocks that {\sl can} consistently be assigned 3-frames, 
and the last alpha carbon linkage
connecting a gap frame to a non-gap frame is twisted or untwisted based upon the usual criteria for the two adjacent well-defined non-gap frames.  In 
particular, the fatgraph constructed is always connected.
Other examples of gap frames arise from breaks along
the backbone as detected by a separation of more than 2.0 Angstroms between atoms $C_i$ and $N_{i+1}$,
for any $i$.

\subsection{Injectivity results}\label{results}

The database CATH version 3.2.0 \cite{CATH} is a collection ${\mathcal P}_{\rm CATH}$ of 114,215 protein domains, which are uniquely catalogued
by a nine-tuple of natural numbers; this is a hierarchical classification with a ``standard'' representative domain chosen in each class.
Our methods have been applied to the associated PDB and DSSP files so as to produce
corresponding connected fatgraphs $G_{-\infty ,E}(P)$ for each $P\in{\mathcal P}_{\rm CATH}$ and various energy thresholds $E<0$.  We have concentrated here just on the question of finding tuples of robust invariants
that uniquely determine the domain $P$  among all the domains in ${\mathcal P}_{\rm CATH}$, or the standard representatives of all the classes at some level, and this section simply presents these empirical ``injectivity''  results.

{\tiny
\begin{table}[ht] \label{tableone}
\caption{Exceptions to injectivity in Result~\ref{results1}}
\centering
\begin{tabular}{c c }
\hline\hline
{\rm Invariants} & {\rm CATH~domains} \\ [0.5ex]
\hline
(26.5,80,23.5,66,21.5,58,16.5,44,11.5,18,5.0,4,3.0,2)&2.60.120.20.4.3.1.2.2 and 2.60.120.20.4.3.1.1.$n$,\\
& for $2\leq n\leq 24$ and $n\neq 3,4,10,11,12,14$\\
(36.5,81,32.5,72,31.5,66,29.0,56,23.5,34,14.0,12,2.0,2)&2.70.98.10.2.1.1.$n$.1,~for~$3\leq n\leq 17$~and~$n\neq 9$\\ 
(34.5,84,32.5,71,31.5,69,29.5,56,22.5,41,14.0,20,5.0,9)&2.70.98.10.2.1.1.$n$.1,~for~$19\leq n\leq 33$~and~$n\neq 26,30$\\
(20.5,89,17.5,82,14.0,66,8.5,48,6.5,25,3.0,12,1.0,3)&3.20.20.70.69.3.1.$n$.1,~for~$4\leq n\leq 10$~or~$n= 12,15,17$\\
(41.0,99,30.5,76,25.5,51,14.0,31,8.0,19,5.5,9,0.5,3)&3.75.10.10.1.2.2.$n$.1,~for~$1\leq n\leq 6$~or~$n=8,11$\\
(20.5,89,17.5,82,14.0,67,8.5,48,6.5,25,3.0,12,1.0,3)&3.20.20.70.69.3.1.$n$.1,~for~$n=1,2,3,13,14,16$\\
(8.0,71,6.0,63,5.5,55,5.0,43,3.0,17,0.0,4,0.0,1)&3.40.50.510.1.1.1.1.$m.n$,~for~$m.n=1.1,1.3,2.3,3.1$\\
(19.5,68,16.5,54,12.5,48,12.5,28,7.5,18,1.0,11,1.0,4)&3.90.70.10.3.2.1.$m.n$,~for $m.n=2.15,4.1,5.1,8.1,9.1$\\
(4.0,96,4.0,91,2.5,86,1.0,64,0.0,18,0.0,1,0.0,1)&1.10.490.10.5.1.1.$m.n$,~for~$m.n=1.52,1.53,28.1,28.2$\\
(4.0,102,3.0,93,2.5,84,1.0,58,0.0,22,0.0,2,0.0,1)&1.10.490.10.4.1.1.$m.n$,~for~$m.n=1.54,1.55,2.17,2.18$\\
(7.5,38,7.0,33,5.0,32,2.5,20,1.5,10,1.0,5,0.5,4)&2.60.40.10.2.1.1.$m.n$,~for~$m.n=1.258,1.259,7.23,7.24$\\
(1.0,29,0.5,29,0.5,27,0.0,20,0.0,11,0.0,5,0.0,2)&4.10.220.20.1.1.2.$n$.1,~for~$n=1,2,3$\\
(4.5,169,4.0,157,2.5,145,2.0,113,1.0,59,0.5,10,0.5,1)&1.20.1070.10.1.1.1.$m.n$,~for~$m.n=1.12,1.21,9.1$\\
(34.0,83,32.0,76,31.5,71,29.0,55,20.0,41,16.0,26,5.0,7)&2.70.98.10.2.1.1.$n$.1,~for~$n=42,44,46$\\
(36.0,136,34.0,123,32.0,114,23.5,85,8.0,40,2.0,15,0.0,5)&3.20.20.140.22.1.1.$n$.1,~for~$n=2,3,4$\\
(0.0,11,0.0,6,0.0,4,0.0,2,0.0,2,0.0,1,0.0,1)&2.10.210.10.1.1.1.1.1~and~1.10.8.10.13.1.1.1.2\\ 
(0.5,32,0.0,30,0.0,29,0.0,20,0.0,8,0.0,2,0.0,1)&4.10.220.20.1.1.1.$n$.1,~for~$n=13,15$\\
(0.5,97,0.5,94,0.5,85,0.5,65,0.5,25,0.5,5,0.0,1)&1.20.1500.10.3.1.1.$n$.1,~for~$n=1,2$\\
(1.0,21,1.0,17,0.5,15,0.5,14,0.5,9,0.0,3,0.0,2)&2.10.69.10.3.2.2.1.1~and~2.10.69.10.3.2.5.1.1\\
(1.5,42,1.5,42,1.5,39,0.5,32,0.0,16,0.0,5,0.0,1)&1.20.1280.10.1.1.1.$m.n$,~for~$m.n=1.1,2.47$\\
(1.5,43,1.5,42,1.5,38,0.5,32,0.0,17,0.0,5,0.0,1)&1.20.1280.10.1.1.1.$m.n$,~for~$m.n=1.2,2.48$\\
(3.0,21,3.0,18,3.0,15,3.0,13,1.5,9,0.5,3,0.5,1)&4.10.410.10.1.1.3.$n$.2, for $n=4,7$\\
(3.0,21,3.0,18,3.0,16,3.0,13,2.0,8,0.5,3,0.0,1)&4.10.410.10.1.1.3.$n$.1, for $n=5,8$\\
(4.5,8,4.5,8,3.5,7,2.5,5,1.0,4,0.0,2,0.0,1)&2.10.25.10.20.2.1.$n$.1, for $n=1,2$\\
(4.5,35,3.5,34,3.0,29,1.5,23,1.0,14,0.0,6,0.0,1)&1.10.1200.30.1.1.2.$m.n$, for $m.n=1.3,4.1$\\
(4.5,51,4.5,42,4.5,32,4.0,20,3.5,15,2.5,5,1.0,4)&3.30.70.270.4.1.1.$m.n$, for $m.n=1.185,2.1$\\
(5.5,48,4.5,40,4.0,32,3.0,18,1.0,12,0.0,10,0.0,4)&1.10.238.10.3.1.2.$n$.1, for $n=5,6$\\
(6.0,42,5.5,36,5.5,31,5.0,29,4.5,15,2.0,1,0.0,1)&2.40.70.10.3.1.1.$m.n$, for $m.n=5.6,6.10$\\
(6.0,44,5.5,39,4.5,30,4.5,23,3.5,14,1.0,6,1.0,2)&1.10.760.10.6.1.1.$n$.1, for $n=1,25$\\
(6.5, 32,6.0,30,5.5,27,3.5,28,2.5,16,1.5,4,0.0,1)&2.30.30.140.3.1.1.$m.n$, for $m.n=1.3,2.1$\\ 
(6.5,44,4.5,41,4.5,35,4.0,25,3.5,13,2.5,7,0.5,4)&3.30.70.270.7.1.2.1.1 and 3.30.70.270.2.1.5.5.2\\
(6.5,57,6.0,52,6.0,52,5.5,42,3.5,25,2.5,7,0.5,1)&3.30.365.10.4.1.1.$m.n$, for $m.n=1.1,2.2$\\
(7.0,65,7.0,64,6.5,60,3.0,54,2.0,28,0.5,5,0.0,1)&1.10.1040.10.4.1.1.$n$.1, for $n=1,2$\\
(7.5,71,6.5,63,4.5,57,4.5,41,2.0,19,2.0,6,0.0,3)&3.30.1330.10.1.1.1.$n$.1, for $n=2,4$\\
(7.5,72,5.5,64,5.0,56,5.0,43,3.0,17,0.0,4,0.0,1)&3.40.50.510.1.1.1.$m.n$, for $m.n=2.4,3.2$\\
(8.0,65,8.0,57,7.5,50,6.0,35,3.5,24,1.0,8,0.0,1)&3.30.1330.10.1.1.1.$n$.1, for $n=3,5$\\
(8.5,35,8.0,33,7.5,31,6.0,26,4.0,17,3.5,4,0.5,2)&2.30.30.140.3.1.1.$m.n$, for $m.n=1.4,2.2$\\ 
(8.5,69,7.5,62,6.5,56,5.5,45,5.5,24,3.5,3,0.0,1)&3.40.47.10.8.1.1.$n$.4, for $n=2,6$\\
(8.5,70,7.5,62,7.0,56,6.0,40,4.5,20,2.5,2,0.0,1)&3.40.47.10.8.1.1.$n$.1, for $n=2,6$\\
(9.0,68,8.0,60,6.5,53,6.0,40,5.0,12,1.5,1,0.5,1)&3.40.47.10.8.1.1.$n$.8, for $n=2,6$\\
(9.0,69,7.5,63,6.5,54,5.5,43,4.5,14,1.0,2,0.0,1)&3.40.47.10.8.1.1.$n$.6, for $n=2,6$\\
(9.0,70,7.5,63,6.5,55,6.0,43,5.0,19,2.0,1,0.0,1)&3.40.47.10.8.1.1.$n$.2, for $n=2,6$\\
(9.5,67,7.5,60,5.5,52,5.0,41,4.0,12,2.0,3,0.0,1)&3.40.47.10.8.1.1.$n$.7, for $n=2,6$\\
(9.5,67,8.0,61,6.0,54,5.0,43,5.0,19,2.0,3,0.0,1)&3.40.47.10.8.1.1.$n$.3, for $n=2,6$\\
(9.5,68,8.0,62,7.5,52,6.0,37,4.0,16,1.5,2,0.0,1)&3.40.47.10.8.1.1.$n$.5, for $n=2,6$\\
(9.5,71,6.5,62,6.0,52,3.5,43,2.5,27,2.0,8,1.5,5)&3.40.420.10.2.2.4.$n$.1, for $n=1,2$\\
(10.5,36,10.5,32,9.0,28,7.5,24,4.5,14,0.0,7,0.0,3)&3.10.20.30.6.1.1.$n$.1, for $n=2,4$\\
(10.5,58,10.0,49,10.0,47,8.5,33,7.0,15,3.5,4,0.5,2)&3.10.310.10.6.1.2.$n$.1, for $n=1,2$\\
(13.5,73,13.5,65,11.5,60,10.5,45,7.5,22,2.5,7,0.0,2)&3.40.50.720.82.1.1.$n$.1, for $n=4,9$\\
(13.5,74,13.5,67,11.5,64,10.5,37,6.5,14,1.5,7,0.0,2)&3.40.50.720.82.1.1.$n$.1, for $n=2,6$\\
(14.0,49,14.0,44,13.0,43,13.0,39,9.5,17,3.0,5,0.0,1)&3.30.1330.40.2.1.1.$n$.1, for $n=1,3$\\
(14.0,58,12.0,52,11.5,47,10.5,33,7.5,14,3.0,8,0.0,5)&3.10.310.10.8.1.1.$n$.1, for $n=6,7$\\
(17.5,79,15.0,64,12.5,54,9.5,38,6.5,21,4.0,6,1.5,2)&2.60.120.20.9.3.1.$m.n$, for $m.n=1.15, 6.1$\\ 
(18.5,81,14.5,65,13.5,55,10.0,40,7.0,24,3.5,8,1.5,2)&2.60.120.20.9.3.1.$m.n$, for $m.n=1.18,6.2$\\
(19.0,20,18.0,21,16.0,17,9.5,18,7.0,10,2.0,4,0.5,2)&2.60.30.10.2.1.1.$n$.1, for $n=7,9$\\
(19.0,55,18.0,50,17.0,45,14.0,34,8.0,18,2.0,5,0.0,1)&3.40.50.720.63.1.$n$.1.1, for $n=1,2$\\
(19.5,149,19.5,137,18.0,124,12.5,97,7.5,50,1.5,7,0.0,1)&3.20.20.110.1.1.3.$n$.1, for $n=11,13$\\
(19.5,180,18.5,161,16.0,135,14.0,77,10.0,28,1.0,8,0.0,2)&3.20.20.70.55.2.1.$m.n$, for $m.n=5.8,7.4$\\
(19.5,185,15.5,163,11.5,130,11.5,82,6.0,42,3.5,10,0.0,1)&3.20.20.70.55.2.1.$m.n$, for $m.n=5.5,7.1$\\
(20.0,43,18.5,38,15.5,31,13.5,22,9.5,14,7.0,6,2.0,4)&3.90.650.10.1.1.1.$n$.1, for $n=3,5$\\
(20.5,61,18.5,51,16.5,47,15.5,31,10.5,20,5.5,5,0.0,4)&2.60.90.10.1.3.1.$n$.1, for $n=1,3$\\ 
(21.5,46,17.0,38,15.0,33,13.5,23,9.5,14,4.5,5,2.0,2)&3.90.650.10.1.1.1.$n$.1, for $n=2,4$\\
(22.0,178,19.0,157,18.0,129,15.0,86,9.5,30,2.0,6,0.0,1)&3.20.20.70.55.2.1.$m.n$, for $m.n=5.6,7.2$\\
(23.0,178,20.0,160,18.0,134,14.5,82,11.0,34,2.0,9,0.0,2)&3.20.20.70.55.2.1.$m.n$, for $m.n=5.7,7.3$\\
(24.0,274,19.5,257,16.0,228,13.0,176,10.0,90,1.0,22,0.0,2)&1.10.620.20.6.1.1.$m.n$, for $m.n=1.2,2.48$\\
(26.5,171,24.0,151,20.5,134,16.5,105,12.5,52,3.0,16,1.0,1)&3.40.718.10.4.6.1.$m.n$, for $m.n=1.4, 3.2$\\
(27.5,180,22.0,160,19.5,141,16.5,105,10.5,51,6.0,12,0.5,3)&3.40.718.10.4.6.1.$m.n$, for $m.n=1.3, 3.1$\\
(36.0,102,28.5,94,26.0,81,20.0,58,12.5,27,6.5,9,2.0,2)&3.50.50.60.55.1.1.$n$.1, for $n=7,9$\\ 
(36.5,81,32.5,72,31.5,66,29.0,56,24.5,33,14.0,12,2.0,2)&2.70.98.10.2.1.1.$n$.1, for $n=9,18$\\  
(36.5,145,34.0,130,27.5,124,25.0,92,15.5,37,3.5,6,0.5,1)&3.20.20.70.72.1.1.$m.n$, for $m.n=3.8,5.4$\\
(36.5,145,34.0,131,28.5,123,25.5,96,17.0,41,5.0,6,0.5,1)&3.20.20.70.72.1.1.$m.n$, for $m.n=3.6,5.2$\\
(38.5,141,36.0,126,30.5,117,27.0,90,19.0,39,4.5,6,0.5,1)&3.20.20.70.72.1.1.$m.n$, for $m.n=3.7,5.3$\\
(39.0,142,35.5,127,30.0,119,26.5,92,16.5,37,5.5,5,1.0,1)&3.20.20.70.72.1.1.$m.n$, for $m.n=3.5,5.1$\\
(41.0,99,30.5,76,25.5,51,14.0,30,8.0,19,5.5,9,0.5,3)&3.75.10.10.1.2.2.$n$.1, for $n=7,10$\\ 

\hline
\end{tabular}
%\label{table:nonlin}
\end{table}
}

Our first results rely only on the most basic of robust invariants which depend only on the topological type of the surface, namely, the modified genus 
$g^*_E(P)$ and the number $r_E(P)$ of boundary components of $F(G_{-\infty,E}(P))$.

\begin{result} \label{results1}
The 14 numbers
$(g^*_E(P),r_E(P))$, with $E=-0.5(1+t)$, for integral $0\leq t\leq 6$, 
uniquely determine the primary structure of each $P\in{\mathcal P}_{\rm CATH}$
except for the special cases given in Table 5.1. %\ref{tableone}.  
In particular, these 14 numbers uniquely determine the depth 7 classes {\rm (CATHSOL)} except for the four following
special cases: {\rm 3.40.50.720.63.1.1.1.1 {\it and } 3.40.50.720.63.1.2.1.1;
 3.30.70.270.7.1.2.1.1 {\it and} 3.30.70.270.2.1.5.5.2; 
2.10.210.10.1.1.1.1.1 {\it and} 1.10.8.10.13.1.1.1.2;
2.10.69.10.3.2.2.1.1 {\it and} 2.10.69.10.3.2.5.1.1.}.
\end{result}

{\tiny
\begin{table}[ht] \label{tabletwo}
\caption{Exceptions to injectivity in Result~\ref{results2}}
\centering
\begin{tabular}{c c }
\hline\hline
{\rm Invariants} & {\rm CATH~domains} \\ [0.5ex]
\hline
(49,45,46,0.0,4.0,0,0,0,0,46)&1.20.5.190.1.1.2.1.4, 1.20.5.530.1.1.1.1.2, 1.20.5.170.1.1.2.1.1\\
(56,51,52,0.0,4.0,0,0,0,0,53)&1.20.5.190.1.1.3.1.1, 1.20.5.500.1.1.1.1.3, 1.20.5.170.9.1.1.1.1\\
(42,38,39,0.0,4.0,0,0,0,0,39)&1.20.5.190.1.1.3.2.1, 1.20.5.170.3.1.1.1.12\\
(46,31,30,1.0,5.0,2,3,0,2,39)&1.10.60.10.3.1.1.1.2, 1.10.287.680.1.1.1.1.16\\
(49,43,44,0.0,4.1,0,0,0,0,46)&1.20.5.300.2.1.1.1.7, 1.20.5.170.2.2.1.1.6\\
(49,25,24,1.0,6.0,6,3,1,5,35)&1.10.10.60.32.1.1.1.42, 4.10.51.10.1.1.1.1.25\\
(50,45,46,0.0,4.0,0,0,0,0,47)&1.20.5.80.2.1.1.2.2, 1.20.5.170.2.2.1.1.2\\
(52,48,49,0.0,4.0,0,0,0,0,49)&1.20.5.530.1.1.1.1.1, 1.20.5.170.2.1.1.1.2\\
(52,32,33,0.0,5.0,6,1,2,3,40)&4.10.220.20.1.1.1.1.1, 1.20.5.810.3.1.1.7.1\\
(53,30,27,2.0,6.0,5,6,1,4,41)&1.10.1220.10.3.1.3.1.3, 1.10.890.20.1.1.1.1.3\\
(59,55,56,0.0,4.0,0,0,0,0,56)&1.20.5.500.1.1.1.1.2, 1.20.5.170.10.1.1.3.1\\
(60,56,57,0.0,4.0,0,0,0,0,57)&1.20.5.500.1.1.1.1.1, 1.20.5.170.10.1.1.3.2\\
(62,58,59,0.0,4.0,0,0,0,0,59)&1.20.5.170.6.1.1.2.1, 1.20.5.110.6.1.1.2.3\\
(64,58,59,0.0,4.1,0,0,0,0,61)&1.20.5.300.1.1.1.1.2, 1.20.5.170.6.1.1.1.8\\
(65,37,35,1.5,5.7,9,5,2,7,46)&1.10.8.200.1.1.1.2.1, 1.10.2030.10.1.1.1.1.8\\
(72,48,46,1.5,5.1,7,3,2,5,57)&1.10.40.30.1.1.2.1.6, 1.10.220.10.8.1.1.1.2\\
(79,75,76,0.0,4.0,0,0,0,0,76)&1.20.5.170.16.1.1.1.5, 1.20.5.110.7.1.1.2.1\\
(88,60,53,4.0,5.5,10,11,4,6,69)&1.10.238.10.9.2.1.1.10, 1.10.288.10.2.1.1.1.1\\
(95,54,42,6.5,7.0,38,23,26,11,43)&3.30.1050.10.5.1.1.1.6, 3.30.1490.70.4.1.1.1.2\\
\hline
\end{tabular}
%\label{table:nonlin}
\end{table}
}

{\tiny
\begin{table}[ht] \label{tablethree}
\caption{Exceptions to injectivity in Result~\ref{results3}, where ${\rm N}^k$ denotes $k\geq 1$ consecutive N}
\centering
\begin{tabular}{c c }
\hline\hline
{\rm Flip Sequence} & {\rm CATH~domains} \\ [0.5ex]
\hline
$N^{19}$&1.20.5.460.1.1.1.6.1,~1.20.5.110.15.1.1.1.1\\
$N^{27}$&1.20.5.800.1.1.2.1.1,~1.10.10.380.1.1.1.1.1\\
$N^{29}$&1.20.5.140.3.1.1.1.1,~1.20.5.420.5.1.1.1.1,~1.20.5.170.18.1.1.1.1\\
$N^{30}$&1.20.5.700.1.1.1.1.1,~1.20.5.100.2.1.1.1.1\\
$N^{32}$&1.20.5.770.1.1.1.1.1,~1.20.5.700.1.1.1.1.3\\
$N^{37}$&1.20.5.40.1.1.2.1.6,~1.20.5.80.2.1.1.2.5\\
$N^{38}$&1.20.5.440.1.1.1.1.1,~4.10.810.10.1.1.1.1.1,~1.20.5.170.8.1.1.1.5\\
$N^{40}$&1.20.5.190.1.1.3.2.1,~1.20.5.170.3.1.1.1.12\\
$N^{42}$&1.20.5.430.1.1.2.1.3,~1.20.5.80.2.1.1.1.3,~1.20.5.490.1.1.1.1.1\\
$N^{43}$&1.20.5.240.1.2.1.1.1,~1.10.930.10.1.1.2.1.2,~1.20.5.170.3.1.1.1.1\\
$N^{44}$&1.20.5.230.1.1.1.1.1,~1.20.5.80.1.1.1.1.2\\
$N^{45}$&1.20.5.190.1.1.2.1.5,~1.20.5.300.2.1.1.1.12,~1.20.5.170.14.1.1.1.1\\
$N^{46}$&1.20.5.300.2.1.1.1.9,~1.10.287.300.1.1.1.1.1\\
$N^{47}$&1.20.5.190.1.1.2.1.4,~1.20.5.530.1.1.1.1.2,~1.20.5.300.2.1.1.1.7,~1.20.5.170.1.1.2.1.1\\
$N^{48}$&1.20.5.190.1.1.1.1.2,~1.20.5.80.2.1.1.2.1,~1.20.5.300.2.1.1.1.1,~1.20.5.170.2.2.1.1.1\\
$N^{49}$&1.20.5.190.1.1.2.1.1,~1.20.5.170.2.2.1.1.11,~1.20.5.110.2.1.1.1.3\\
$N^{50}$&1.20.5.290.1.1.1.1.1,~1.20.5.530.1.1.1.1.1,~1.20.5.170.2.1.1.1.2,~1.20.5.110.14.1.1.1.1\\
$N^{51}$&1.20.5.190.1.1.5.1.1,~1.20.5.370.2.1.2.1.1,~1.20.5.170.10.1.1.1.1\\
$N^{52}$&1.10.287.750.1.1.8.1.1,~1.20.5.170.2.2.1.2.2,~1.20.5.110.11.1.1.1.1\\
$N^{53}$&1.20.5.170.2.2.1.2.1,~1.20.5.110.10.1.1.1.1\\
$N^{54}$&1.20.5.190.1.1.3.1.1,~1.20.5.500.1.1.1.1.3,~1.20.5.170.4.1.1.1.1\\
$N^{56}$&1.20.5.300.1.2.1.1.2,~1.20.5.110.5.1.1.1.2\\
$N^{57}$&1.20.5.500.1.1.1.1.2,~1.20.5.170.10.1.1.3.1,~1.10.287.130.2.1.1.1.6\\
$N^{58}$&1.20.5.390.1.1.1.1.1,~1.20.5.500.1.1.1.1.1,~1.20.5.170.10.1.1.3.2,~1.20.5.110.8.1.1.1.1\\
$N^{59}$&1.20.5.620.1.1.1.1.1,~1.10.287.230.1.1.1.1.2,~1.20.5.170.4.2.1.1.1,~1.20.5.110.5.1.1.1.1\\
$N^{60}$&1.20.5.300.1.1.1.1.1,~1.20.5.170.4.1.1.2.2,~1.20.5.110.6.1.1.2.3\\
$N^{61}$&1.10.287.210.2.2.1.8.1,~1.20.5.170.6.1.1.1.11,~1.20.5.110.3.1.1.1.1\\
$N^{62}$&1.20.5.300.1.1.1.1.2,~1.20.5.170.6.1.1.1.8,~1.20.5.110.4.1.1.1.1\\
$N^{63}$&1.20.5.500.1.1.1.1.4,~1.20.5.170.5.1.1.1.1\\
$N^{65}$&1.10.1440.10.1.1.1.1.1,~1.20.5.170.5.1.1.1.2,~1.2.5.110.6.1.1.1.1\\
$N^{66}$&1.20.5.730.1.1.1.1.1,~1.20.5.170.6.1.1.1.3,~1.20.5.110.2.1.1.1.1\\
$N^{71}$&1.20.5.400.1.1.1.1.1,~1.10.287.210.2.2.1.4.4,~1.20.5.110.6.1.2.2.2\\
$N^{72}$&1.10.287.210.2.2.1.4.3,~1.20.5.170.16.1.1.1.3,~1.20.5.110.6.1.2.2.3\\
$N^{75}$&1.20.5.340.1.1.1.1.4,~1.20.5.110.7.1.1.4.3\\
$N^{76}$&1.10.287.210.7.1.1.1.1,~1.20.5.170.16.1.1.1.4\\
$N^{77}$&1.20.20.10.1.1.1.1.3,~1.20.5.340.1.1.1.1.3,~1.20.5.170.16.1.1.1.5,~1.20.5.110.7.1.1.2.1\\
$N^{28}F N^{25}$&1.10.287.660.1.1.1.2.1,~1.10.287.230.1.1.2.1.5,~1.10.287.750.1.1.6.1.1\\
$N^2F N^{61}$&1.20.5.170.5.1.1.2.1,~1.20.5.110.6.1.1.2.1\\
$N^{27}F N^{26}$&1.10.287.230.1.1.1.4.1,~1.10.287.210.2.1.2.1.3\\
$N^{29}FN^{24}$&1.10.287.230.1.1.2.1.4,~1.10.287.750.1.1.5.1.1\\
$N^{31}FN^{26} $&1.10.287.750.1.1.3.1.1,~1.10.287.210.2.2.1.7.1\\
$N^{34}FNF^2N$&4.10.81.10.2.1.1.1.1,~1.20.5.50.9.1.1.1.8\\
$N^{41}F$&1.20.5.490.1.1.1.1.3,~1.20.1070.10.7.1.1.1.2\\
$N^{43}F$&1.10.10.200.2.2.1.1.1,~1.20.5.170.15.1.1.1.1\\
$N^{50}F$&1.20.5.170.10.1.1.2.1,~1.10.287.190.1.1.1.1.2\\
\hline
\end{tabular}
\label{table,nonlin}
\end{table}
}

The next injectivity result relies upon several robust invariants of the fatgraph.

\begin{result} \label{results2}
For any polypeptide structure $P$ and energy threshold $E<0$, consider the 10 numbers given by: the number of residues of $P$,
the number of hydrogen bonds of $P$ with energy at most $E$, $r_E(P)$, $g^*_E(P)$, the mean of the peptide length spectrum to one significant digit,
the number of twisted alpha carbon linkages of $G_{-\infty ,E}(P)$, the number of twisted edges of $G_{-\infty ,E}(P)$ corresponding to hydrogen bonds, the 
respective number
of pairs {\rm FF},  {\rm FN}, and {\rm NN} occurring in the flip sequence.  These numbers for the single energy level $E=-0.5$ uniquely determine
the standard representatives of ${\mathcal P}_{\rm CATH}$ classes at depth four {\rm (CATH)}  except for the
19 exceptions enumerated in Table \ref{tabletwo}
\end{result}

Our final injectivity result relies only on the model of the backbone, namely, on the flip sequence.

\begin{result}\label{results3} The flip sequence nearly uniquely determines elements of ${\mathcal P}_{\rm CATH}$ with the 45 exceptions
enumerated in Table 5.3. %\ref{tablethree}
\end{result}

We regard Results \ref{results1} to \ref{results3} as  topological classifications of protein domains in the spirit of  topology determining geometry as is familiar from rigidity results for three-dimensional manifolds for example.

\section{Closing remarks}\label{closing}

The fatgraph corresponding to a polypeptide structure defined here, and its generalizations discussed in Section~\ref{generalmodel}, is based on the intrinsic 
geometry of a protein at equilibrium.
We believe that we have just scratched the surface of defining meaningful protein descriptors derived from robust invariants of these fatgraphs
in this paper, whose primary intent is simply to introduce these methods.  Further applications are either ongoing or anticipated, and we briefly discuss
aspects of these various projects in this closing section.

Recall from Section \ref{generalmodel} that rotamer fatgraphs arise from our basic fatgraph model of a polypeptide structure by refining the simplest 
discretization of the backbone graph connection.
Such a rotamer fatgraph or invariants of it may be assigned to the subsequence of a protein corresponding to a turn or coil in order to give a new 
classification of these structural elements.
Construction \ref{constructfatgraph} associates matrices to hydrogen bonds thus providing new tools for their analysis, for example, discretizations likewise 
providing new classifications of hydrogen bonds.

More generally, the fatgraph or rotamer fatgraph of a protein or protein domain and robust invariants of it provide new descriptors
which can be used to refine existing structural classifications.  A key attribute of these new descriptors, as
exemplified by the injectivity results in Section~\ref{results}, is that they are automatically computable from PDB files without the need
for human interpretation into the usual architectural motifs.
In a similar vein, \cite{Roegen} associates protein descriptors inspired by quantum invariants of links, which are different from the quantum invariants
proposed in Section \ref{generalmodel}, and proves injectivity results analogous to those in
Section \ref{results}.  In contrast to \cite{Roegen} where the geometric or topological meaning
of the descriptors is unclear, the significance of our descriptors such as those considered in Section \ref{results} is manifest.

The recent paper \cite{Boom}
studies probability densities on the space of conformational angles with applications to structure prediction, and
densities on the Lie group $SO(3)$ can be computed and applied to structure prediction in an analogous manner.
Furthermore, the prediction of corresponding discretizations such as the flip sequence and its rotamer analogues from
protein primary structure has already proved interesting.

\appendix

\section{Alternative description of the model}\label{alternative}

There is another representative $K(P)$ of the equivalence class of the fatgraph $G(P)$ associated to a polypeptide structure $P$ which we shall
describe in this appendix.  In some ways, the alternative description is more natural though Corollary~\ref{risrobust} is true but not obvious in this 
formulation.

\begin{figure}[!h]
\begin{center}
\epsffile{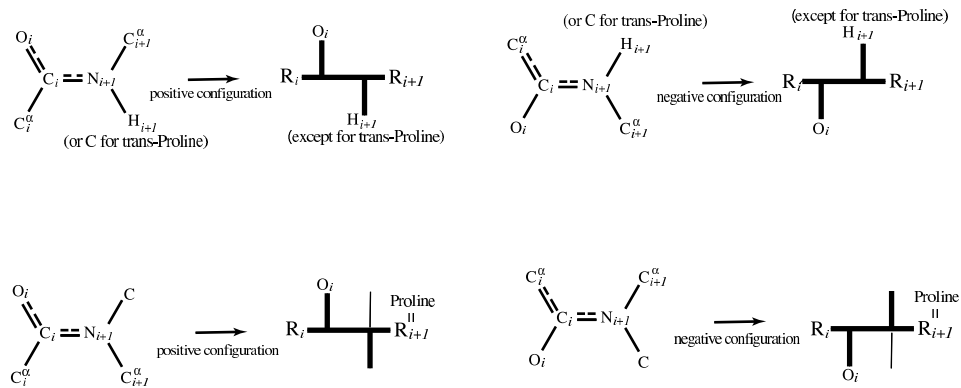}
\caption{Fatgraph building blocks for the alternative model}
\label{fig:buildingblocks}
\end{center}
\end{figure}

The backbone is modeled as the concatenation of fatgraph building blocks, one such building block for each peptide unit.
The two possible building blocks for the $i$th peptide unit are illustrated in Figure~\ref{fig:buildingblocks} and are called the {\it positive} and {\it 
negative configurations} corresponding to whether the oxygen atom
$O_i$ lies to the left or right of the backbone, respectively, when traversed from $N$ to $C$ termini.
The model of the backbone is determined by the sequence of configurations, positive or negative, assigned to the consecutive peptide units and is thus described 
by a word of length $L-1$ in the alphabet $\{ +,-\}$, which is called the {\it plus/minus sequence} of the polypeptide structure.  The untwisted fatgraph $Y(P)$ 
which is an alternative model of the backbone is  constructed from this data by identifying endpoints of the consecutive horizontal segments of the fatgraph 
building blocks in the natural way as before.
There is an arbitrary choice of configuration $c_1=+$ for the first building block as positive.

Suppose recursively that configurations $c_\ell\in\{{+,-}\}$ have been determined for $\ell <i < L$.
The configuration $c_i$ is calculated from the configuration $c_{i-1}$
as follows:
$$c_i =\begin{cases}
+c_{i -1}, &{\rm if}~\vec v_{i-1}\cdot\vec v_i~+~\vec w_{i-1}\cdot \vec w_i >0~{\rm and}~R_i~{\rm is~not~cis-Proline};\\
-c_{i -1}, &{\rm if}~\vec v_{i-1}\cdot\vec v_i~+~\vec w_{i-1}\cdot \vec w_i \leq 0 ~{\rm and}~R_i~{\rm is~not~cis-Proline};\\
-c_{i -1}, &{\rm if}~\vec v_{i-1}\cdot\vec v_i~+~\vec w_{i-1}\cdot \vec w_i \geq 0, ~{\rm and}~R_i~{\rm is~cis-Proline};\\
+c_{i -1}, &{\rm if}~\vec v_{i-1}\cdot\vec v_i~+~\vec w_{i-1}\cdot \vec w_i <0,~{\rm and}~R_i~{\rm is~cis-Proline},\\
\end{cases}$$
completing the construction of the alternative backbone model $Y(P)$.  Notice that the flip sequence uniquely determines the plus/minus sequence and conversely.

As in Construction~\ref{constructfatgraph'}, if $(i,j)\in{\mathcal B}$ in input ii), then we add an edge to $Y(P)$ connecting the short vertical segments 
corresponding to the atoms
$H_{i}$ and $O_{j}$.
To complete the construction of $K(P)$, it remains only to specify which edges of the resulting fatgraph are twisted.
To this end, suppose that  $(i,j)\in{\mathcal B}$ in input ii).  There are corresponding 3-frames
$$\aligned
{\mathcal F}_{i-1}&=(\vec u_{i-1},\vec v_{i-1},\vec w_{i-1}),\\
{\mathcal F}_{j}&=(\vec u_j,\vec v_j,\vec w_j),\\
\endaligned$$
from Construction \ref{3-frames}
and corresponding configurations $c_{i-1}$ and $c_j$ defined above.
An edge corresponding to the hydrogen bond $(i,j)\in{\mathcal B}$ is taken to be
twisted in $K(P)$ if and only if
$c_{i-1}c_j~{\rm sign}(\vec v_{i-1}\cdot \vec v_j + \vec w_{i-1}\cdot \vec w_j)$ is negative.

\begin{figure}[!h]
\begin{center}
\epsffile{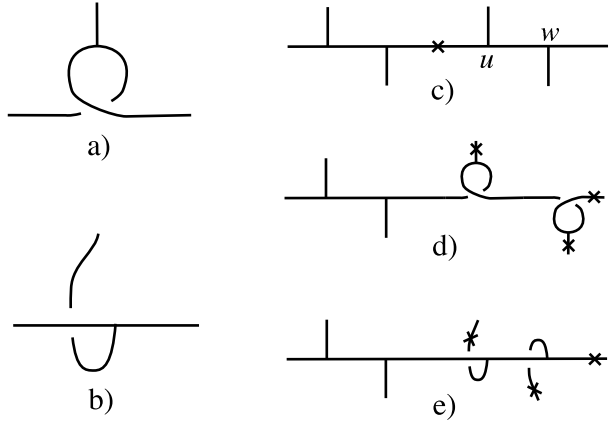}
\caption{Elementary equivalences of fatgraphs}
\label{fig:elementarystrongequivalences}
\end{center}
\end{figure}

The proof that $K(P)$ and $G(P)$ are equivalent depends upon the following simple diagrammatic result.

\begin{lemma}\label{elementarystrongequivalences}
The fatgraphs depicted in Figures~\ref{fig:elementarystrongequivalences}a and \ref{fig:elementarystrongequivalences}b are strongly
equivalent, and the fatgraphs depicted in Figures~\ref{fig:elementarystrongequivalences}c, \ref{fig:elementarystrongequivalences}d, and 
\ref{fig:elementarystrongequivalences}e
are pairwise equivalent.
\end{lemma}

\begin{proof}
The strong equivalence of \ref{fig:elementarystrongequivalences}a and \ref{fig:elementarystrongequivalences}b is proved directly.
Perform vertex flips on the vertices labeled $u,w$ in \ref{fig:elementarystrongequivalences}c and erase pairs of icons $\times$ on common edges to produce 
\ref{fig:elementarystrongequivalences}d, which is strongly equivalent to \ref{fig:elementarystrongequivalences}e
according to the first assertion.
\end{proof}

\begin{proposition} The fatgraphs $G(P)$ and $K(P)$ are equivalent.
\end{proposition}

\begin{proof} The underlying graphs of $G(P)$ and $K(P)$ are isomorphic by construction.  Furthermore, recursive application of 
Lemma~\ref{elementarystrongequivalences} shows that
there is a sequence of vertex flips starting at $T(P)$ and ending at $Y(P)$, so the two backbone models are equivalent by Proposition~\ref{strong}.  We claim 
that
an edge of $G(P)$ representing a hydrogen bond is twisted if and only if it the corresponding edge of
$K(P)$ is twisted, and there are two cases depending upon the parity of the number of twisted alpha carbon linkages of $G(P)$ between the endpoints of such an 
edge.  This number is even, and hence
so too is the number of icons $\times$ on the edge,  if and only if the configurations of fatgraph building blocks in $K(P)$ at these endpoints agree, and the 
claim therefore follows by definition
of twisting in $K(P)$.
\end{proof}

\begin{figure}[!h]
\begin{center}
\epsffile{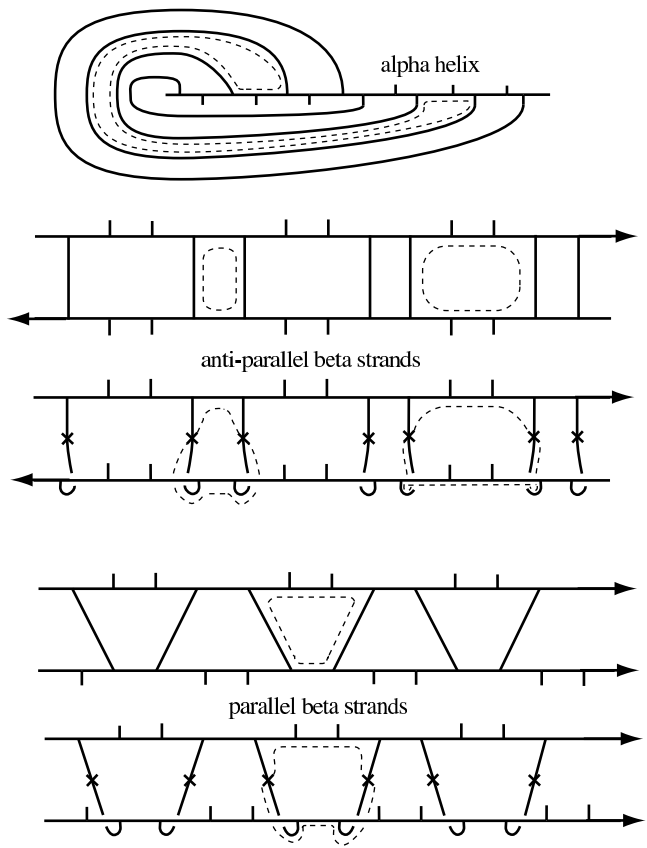}
\caption{Alpha helices and beta strands}
\label{fig:alphanbeta}
\end{center}
\end{figure}

We finally  consider how the standard motifs of protein secondary structure are manifest in our alternative model $K(P)$.
The illustration on the top of Figure~\ref{fig:alphanbeta} depicts our fatgraph model of an {\it alpha helix}, which is defined by the indicated pattern of 
hydrogen bonding.
It is well-known for proteins \cite{Finkelstein} that the plus/minus sequence of an alpha helix is given by
a constant\footnote{This can be seen, for example, from the Ramachandran plot Figure~\ref{fig:rama} or from the direct consideration of 3-frames according to 
Construction \ref{backbone}.} plus/minus sequence
 $+~+ ~+ ~+ ~+$ or $-~-~-~-~-$. Indeed, this is the standard graphical depiction of an alpha helix in the protein literature, but for us, there is the deeper 
 meaning of the figure as a fatgraph rather than simply as a graph in its usual interpretation.  The dotted line indicates a typical boundary component of the 
 corresponding surface.

The second  and fourth illustrations from the top in Figure~\ref{fig:alphanbeta} depict our fatgraph models of an {\it anti-parallel beta strand} and a {\it 
parallel beta strand}, respectively, which are again defined by the indicated pattern of hydrogen bonding and the orientations along the backbone from the $N$
to $C$ termini indicated by the arrows in the figure.  Again, it is well-known for proteins \cite{Finkelstein} that a beta strand, whether parallel or 
anti-parallel, has an alternating\footnotemark[\value{footnote}] plus/minus sequence $+~-~+~-~+$ or $-~+~-~+~-$.
Again, these are the standard graphical depictions of beta strands but now with our enhanced fatgraph interpretation, and the dotted lines indicate typical 
boundary components of the corresponding surface.

Consider the effect of a change of single configuration type in the plus/minus sequence, from $+$ to $-$ or $-$ to $+$, on the backbone between these two 
backbone snippets as depicted in the third and fifth illustrations from the top in Figure~\ref{fig:alphanbeta}.  It follows from the definition of twisting in 
$K(P)$ that the vertical edges corresponding to hydrogen bonds will now be twisted.   The boundary components in the second and fourth illustrations from the 
top persist in the third and fifth illustrations, respectively, in accordance with Corollary~\ref{risrobust}.  Indeed, an odd number of  changes of 
configuration types in the backbone between the two backbone snippets will produce the analogous result, and an even number leaves the figure unchanged.

Let us also clarify a point about anti-parallel beta strands.  It is {\sl not necessarily the case} that the second and third illustrations from the top in 
Figure~\ref{fig:alphanbeta}  accurately depict our fatgraph model of an anti-parallel beta strand: it may happen that our model produces the second figure but 
with twisted edges representing the hydrogen bonds in the strand or the third figure without such twisting.  This is because the determination of twisting in 
$K(P)$ depends upon the sign of  $c c'(\vec v\cdot \vec v' +\vec w\cdot \vec w')$, where $(\vec u,\vec v,\vec w)$ and $(\vec u',\vec v',\vec w')$ are the 
3-frames of the peptide units with configurations $c$ and $c'$ corresponding to the endpoints of the edge. Though the oxygen and hydrogen atoms involved in the 
hydrogen bond are within a few Angstroms, the configurations $c,c'$ may not reflect this, and furthermore, the sign of $c c'(\vec v\cdot \vec v' +\vec w\cdot 
\vec w')$ depends not only on $c$ and $c'$, but also on {\sl both} of $\vec v\cdot \vec v' $ and $\vec w\cdot \vec w'$.  This leads naturally to the notion of
``{untwisted anti-parallel beta strands}'', namely, those for which Figure~\ref{fig:alphanbeta}  is accurate, and ``{twisted anti-parallel beta strands}'', 
those for which it is not.
In contrast, alpha helices and parallel beta strands {\sl are always} represented as in Figure~\ref{fig:alphanbeta}.

In short, the passage from graph to fatgraph enhances the usual graphical depiction of alpha helices and beta strands.  Changes of configuration type away from 
the alpha helices and beta strands leaves undisturbed the boundary components of the surface associated to the fatgraphs which model them.  Furthermore, the 
distinction between twisted and untwisted anti-parallel beta strands is new and depends upon modeling the backbone as a fatgraph rather than merely as a graph.

\bibliographystyle{amsplain}

\end{document}